\documentclass[journal,comsoc,10pt]{IEEEtran}

\usepackage{cite}

\usepackage{graphicx}
\usepackage{xcolor}
\usepackage{tikz}
\usetikzlibrary{arrows,positioning,shapes.geometric,circuits.logic.US}
\usepackage{pgfplots}
\usepackage{multirow}
\usepackage{upgreek}
\usepackage{amssymb}
\usepackage{amsmath}
\usepackage{cases}
\usepackage{url}
\usepackage{pifont}
\usepackage{bm}
\usepackage{amsthm}
\newtheorem{theorem}{Theorem}

\usepackage{tabularx}
\usepackage{booktabs}
\usepackage{filecontents}
\usepackage{subcaption}
\newcolumntype{Y}{>{\centering\arraybackslash}X}

\usepackage{algorithmic}

\usepackage{array}

\newcommand{\fixme}[2]{\ifx&#2&{\leavevmode\color{red}#1}\else{\leavevmode\color{red}FIXME\{}#1{\leavevmode\color{red}\}}\footnote{{\leavevmode\color{red}#2}}\PackageWarning{Fixme}{#1: #2}\fi}

\newcommand{\newstuff}[2]{\ifx&#2&{\leavevmode\color{blue}#1}\else{\leavevmode\color{blue}FIXME\{}#1{\leavevmode\color{blue}\}}\footnote{{\leavevmode\color{blue}#2}}\PackageWarning{Newstuff}{#1: #2}\fi}

\newcommand{\Marco}[2]{\ifx&#2&{\leavevmode\color{blue}#1}\else{\leavevmode\color{blue}FIXME\{}#1{\leavevmode\color{blue}\}}\footnote{{\leavevmode\color{blue}#2}}\PackageWarning{Marco}{#1: #2}\fi}

\DeclareMathOperator*{\sgn}{sgn}

\DeclareMathOperator*{\arctanh}{arctanh}

\title{Rate-Flexible Fast Polar Decoders}

\author{
		Seyyed~Ali~Hashemi,
        Carlo~Condo,
        Marco~Mondelli,
        Warren~J.~Gross
\thanks{S.~A.~Hashemi was with the Department of Electrical and Computer Engineering, McGill University, Montr\'eal, QC H3A 0G4, Canada. He is now with the Department of Electrical Engineering, Stanford University, Stanford, CA 94305, USA (email: ahashemi@stanford.edu).}% <-this % stops a space
\thanks{C.~Condo was with the Department of Electrical and Computer Engineering, McGill University, Montr\'eal, QC H3A 0G4, Canada. He is now with Huawei Technologies France, 92100 Boulogne-Billancourt, France (e-mail:
carlo.condo@huawei.com).}% <-this % stops a space
\thanks{M.~Mondelli is with the Department of Electrical Engineering, Stanford University, Stanford, CA 94305, USA (email: mondelli@stanford.edu).}% <-this % stops a space
\thanks{W.~J.~Gross is with the Department of Electrical and Computer Engineering, McGill University, Montr\'eal, QC H3A 0G4, Canada (e-mail: warren.gross@mcgill.ca).}% <-this % stops a space
}

\begin{document}

\maketitle
\begin{abstract}
Polar codes have gained extensive attention during the past few years and recently they have been selected for the next generation of wireless communications standards (5G). Successive-cancellation-based (SC-based) decoders, such as SC list (SCL) and SC flip (SCF), provide a reasonable error performance for polar codes at the cost of low decoding speed. Fast SC-based decoders, such as Fast-SSC, Fast-SSCL, and Fast-SSCF, identify the special constituent codes in a polar code graph off-line, produce a list of operations, store the list in memory, and feed the list to the decoder to decode the constituent codes in order efficiently, thus increasing the decoding speed. However, the list of operations is dependent on the code rate and as the rate changes, a new list is produced, making fast SC-based decoders not rate-flexible. In this paper, we propose a completely rate-flexible fast SC-based decoder by creating the list of operations directly in hardware, with low implementation complexity. We further propose a hardware architecture implementing the proposed method and show that the area occupation of the rate-flexible fast SC-based decoder in this paper is only $38\%$ of the total area of the memory-based base-line decoder when 5G code rates are supported.
\end{abstract}


\IEEEpeerreviewmaketitle

\section{Introduction} \label{sec:intro}

Polar codes are a family of channel codes which can provably achieve the capacity of a binary memoryless symmetric (BMS) channel with the low-complexity successive-cancellation (SC) decoding algorithm \cite{arikan}. However, this capacity-achieving property under SC decoding only occurs as the code length tends towards infinity. For practical values of code length, SC decoding fails to provide a reasonable error-correction performance.

In order to improve the error-correction performance of SC decoding, SC list (SCL) \cite{tal_list} and SC flip (SCF) \cite{afisiadis} decoders run multiple SC decoders in parallel and in series, respectively. Therefore, SCL improves the error-correction performance of SC at the cost of higher area occupation when implemented on hardware, while SCF improves the error-correction performance of SC at the cost of higher latency and lower throughput. With this error-correction performance improvement, polar codes were selected as a channel coding scheme for the enhanced mobile broadband (eMBB) control channel in the next generation of wireless communications standard (5G).

SC-based decoding algorithms such as SC, SCL, and SCF, suffer from high latency and low throughput when implemented on hardware. This is due to the serial nature of SC decoding in which the decoding proceeds bit by bit. In order to address this issue, polar codes where shown to be a concatenation of smaller constituent codes which can be decoded in parallel \cite{alamdar,sarkis}. These constituent codes are shown to add small implementation complexity overhead while keeping the error-correction performance of SC unchanged. In \cite{hanif}, more constituent codes were identified and low-complexity parallel decoders were designed to increase the throughput of SC decoders even further. It was shown in \cite{hashemi_SSCL_TCASI,hashemi_FSSCL_TSP} that the constituent codes can be decoded efficiently under SCL decoding while keeping the error-correction performance of SCL decoder unaltered. The same approach was applied to the SCF decoder in \cite{giardFlip}.

The construction of polar codes is based on the identification of reliable bit-channels through which information bits are transmitted. The remaining bit-channels carry fix values and are called frozen bits. The location of the frozen bits and of the information bits is known to the encoder and the decoder. In SC-based decoders, the frozen and information bit sequence can be either stored in a memory, or computed on-line given the bit-channel relative reliability vector and desired code rate, as proposed in \cite{Condo_SIPS}. In fact, the latter approach is significantly more efficient in case of multi-code decoders, and is facilitated by nested reliability vectors as those selected for the 5G eMBB control channel \cite{3gpp_polarSequence}. Therefore, in 5G, the polar encoder and decoder are provided with a vector of bit indices in descending reliability order and an information length $K$, from which the encoder and the decoder should extract the frozen/information bit sequence. It should be noted that the number of information bits for polar codes in the 5G eMBB control channel can be any value between $12$ and $1706$ \cite{3gpp_polarInfo}. Thus, the encoder and the decoder should be able to support a vast range of code rates.

Fast SC-based decoders rely on the identification of the type and the length of constituent codes in a polar code. While the calculation of the frozen/information bit sequence is straightforward and can be performed by simply assigning information bits to the first $K$ elements of the reliability vector, the direct calculation of the list of operations for fast SC-based decoders requires complicated controller logic \cite{sarkis}. Therefore, the identification of the type and the length of constituent codes is performed off-line and the decoding order is stored in a dedicated memory as a list of operations \cite{sarkis,hashemi_SSCL_TCASI,hashemi_FSSCL_TSP}. The decoder fetches the list of operations from memory to decode the constituent codes in order one by one. The main drawbacks of the aforementioned fast SC-based decoders are twofold: first, the list of operations requires high memory usage when implemented on hardware. Second, the list of operations is highly dependent on the rate of the polar code and as the rate changes, the list of operations changes too. Therefore, for 5G applications which require the support of multiple rates, multiple lists of operations need to be stored in memory. This in turn increases the hardware implementation overhead and renders fast SC-based decoders not rate-flexible.

In this paper, we propose completely rate-flexible fast SC-based decoders by introducing a method to infer the list of operations directly in hardware by using the bit-channel relative reliability vector and without the need to store it in memory. We show that the type and the length of a constituent code in a polar code can be identified with low hardware implementation complexity, by checking only a few bits of the constituent code. We further show that the list of operations adapts with the rate of the code, allowing the resulting fast SC-based decoder to be completely rate-flexible. We design and implement a hardware architecture for the proposed decoder and show that the memory required to store the list of operations can be completely removed, resulting in significantly lower decoder area occupation.

The remainder of this paper is organized as follows: Section~II reviews polar codes, SC-based decoding algorithms, and their fast counterparts. We propose the rate-flexible fast decoder for polar codes in Section~III. In Section~IV, a hardware architecture to implement the proposed method is introduced. Section~V provides the hardware implementation results and comparisons with state of the art. Finally, conclusions are drawn in Section~VI.

\section{Preliminaries} \label{sec:prel}

\subsection{Polar Codes} \label{sec:prel:polar}

A polar code of length $N=2^n$ that carries $K$ information bits has a rate $R = K/N$ and can be represented as $\mathcal{P}(N,K)$. It can be constructed using a lower-triangular generator matrix $\mathbf{G}$ as
\begin{equation}
\mathbf{x} = \mathbf{u} \mathbf{G} \text{,} \label{eq:polarGen}
\end{equation}
where $\mathbf{x} = \{x_0,x_1,\ldots,x_{N-1}\}$ is the vector of coded bits and $\mathbf{u} = \{u_0,u_1,\ldots,u_{N-1}\}$ is the vector of input bits. The matrix $\mathbf{G} = \mathbf{B}_N \mathbf{F}^{\otimes n}$ where $\mathbf{B}_N$ is the bit-reversal permutation matrix, and $\mathbf{F}^{\otimes n}$ is the $n$-th Kronecker product of the polarizing matrix $\mathbf{F} = \left[\begin{smallmatrix} 1&0\\1&1 \end{smallmatrix}\right]$.

As $N$ goes toward infinity, the polarization phenomenon creates bit-channels that are either completely noisy or completely noiseless and the fraction of noiseless bit-channels equals the channel capacity. For finite practical code lengths, the polarization of bit-channels is incomplete, therefore, there are bit-channels that are partially noisy. In principle, a bit-channel relative reliability vector $\mathbf{v} = \{v_0,v_1,\ldots,v_{N-1}\}$, where $0 \leq v_i < N$, is generated  and fed into the encoder and the decoder based on the polarization phenomenon which shows the rank of each bit-channel. Thus, $\mathbf{v}$ is a vector of integers such that if $v_i < v_j$, then bit-channel $i$ is more reliable (less noisy) than bit-channel $j$. The polar encoding process consists of the classification of the bit-channels in $\mathbf{u}$ into two groups based on $\mathbf{v}$: the $K$ good (more reliable) bit-channels which carry the information bits, and the $N-K$ bad (less reliable) bit-channels that are fixed to a predefined value (usually $0$). This classification can be represented as a sequence of binary values $\mathbf{s} = \{s_0,s_1,\ldots,s_{N-1}\}$ where
\begin{equation}
s_i = \begin{cases} 0 &\mbox{if } v_i \geq K \text{,} \\ 
1 & \mbox{if } v_i < K \text{.} \end{cases}
\end{equation}
More formally, let $W$ be a BMS channel with input alphabet $\mathcal{X}=\{0,1\}$ and output alphabet  
$\mathcal{Y}$, and let $\{W(y \mid x) : x\in \mathcal{X}, y\in \mathcal{Y}\}$ be the transition probabilities. In order to quantify the reliability of the channel $W$, we use the Bhattacharyya parameter $Z(W)\in [0,1]$, that is defined as
\begin{align}\label{eq:Battapar}
& Z(W)= \sum_{y \in \mathcal{Y}} \sqrt{W(y\mid 0)W(y \mid 1)}.
\end{align}
Hence, the good bit-channels are the ones that have the lowest Bhattacharyya parameter.

\subsection{SC-Based Decoding} \label{sec:prel:SCDec}

SC-based decoding algorithms can be represented as a depth-first binary tree search with priority to the left branches as depicted in Fig.~\ref{fig:SCDec}. Two kinds of messages are passed between the nodes in the graph: the soft log-likelihood ratio (LLR) values $\bm{\alpha} = \{\alpha_0,\alpha_1,\ldots,\alpha_{2T-1}\}$ which are passed from a parent node at level $\log_2(2T)=t+1$ to the child nodes at level $\log_2(T)=t$, and the hard bit estimates $\bm{\beta} = \{\beta_0,\beta_1,\ldots,\beta_{2T-1}\}$ which are passed from a child node at level $t$ to a parent node at level $t+1$.

\begin{figure}
  \centering
  \begin{tikzpicture}[scale=1.9, thick]
\newcommand\Triangle[1]{-- ++(0:2*#1) -- ++(120:2*#1) --cycle}
\newcommand\Square[1]{+(-#1,-#1) rectangle +(#1,#1)}

  \draw (0,0) circle [radius=.05];
  
  \draw (-.05,0) -- (.05,0);
  \draw (0,-.05) -- (0,.05);

  \draw (-1.05,-.55) \Triangle{.05};
  \draw (1,-.5) \Square{.05};
%  \draw [gray, very thick] (-1,-.5) circle [radius=.05];
%  \fill [gray, very thick] (1,-.5) circle [radius=.05];

  \draw (-1.5,-1) circle [radius=.05];
  \draw (-.55,-1.05) \Triangle{.05};
  \draw (.5,-1) \Square{.05};
%  \draw [gray, very thick] (-.5,-1) circle [radius=.05];
%  \fill [gray, very thick] (.5,-1) circle [radius=.05];
  \fill (1.5,-1) circle [radius=.05];

  \draw (-1.75,-1.5) circle [radius=.05];
  \draw (-1.25,-1.5) circle [radius=.05];
  \draw (-.75,-1.5) circle [radius=.05];
  \fill (-.25,-1.5) circle [radius=.05];
  \draw (.25,-1.5) circle [radius=.05];
  \fill (.75,-1.5) circle [radius=.05];
  \fill (1.25,-1.5) circle [radius=.05];
  \fill (1.75,-1.5) circle [radius=.05];

  \node at (-1.75,-1.7) {$\hat{u}_0$};
  \node at (-1.25,-1.7) {$\hat{u}_1$};
  \node at (-.75,-1.7) {$\hat{u}_2$};
  \node at (-.25,-1.7) {$\hat{u}_3$};
  \node at (.25,-1.7) {$\hat{u}_4$};
  \node at (.75,-1.7) {$\hat{u}_5$};
  \node at (1.25,-1.7) {$\hat{u}_6$};
  \node at (1.75,-1.7) {$\hat{u}_7$};

  \draw (0,-.05) -- (-1,-.45);
  \draw (0,-.05) -- (1,-.45);

  \draw (-1,-.55) -- (-1.5,-.95);
  \draw (-1,-.55) -- (-.5,-.95);
  \draw (1,-.55) -- (.5,-.95);
  \draw (1,-.55) -- (1.5,-.95);

  \draw (-1.5,-1.05) -- (-1.75,-1.45);
  \draw (-1.5,-1.05) -- (-1.25,-1.45);
  \draw (-.5,-1.05) -- (-.75,-1.45);
  \draw (-.5,-1.05) -- (-.25,-1.45);
  \draw (.5,-1.05) -- (.25,-1.45);
  \draw (.5,-1.05) -- (.75,-1.45);
  \draw (1.5,-1.05) -- (1.25,-1.45);
  \draw (1.5,-1.05) -- (1.75,-1.45);

  \draw [very thin,gray,dashed] (-2,0) -- (2,0);
  \draw [very thin,gray,dashed] (-2,-.5) -- (2,-.5);
  \draw [very thin,gray,dashed] (-2,-1) -- (2,-1);
  \draw [very thin,gray,dashed] (-2,-1.5) -- (2,-1.5);

  \node at (-2.25,0) {$t=3$};
  \node at (-2.25,-.5) {$t=2$};
  \node at (-2.25,-1) {$t=1$};
  \node at (-2.25,-1.5) {$t=0$};

  \draw [->] (-.12,-.05) -- (-1,-.4) node [above=-.1cm,midway,rotate=25] {$\bm{\alpha}$};
  \draw [->] (-.88,-.45) -- (0,-.1) node [below=-.1cm,midway,rotate=25] {$\bm{\beta}$};

  \draw [->] (-1.06,-.55) -- (-1.5,-.9) node [above=-.1cm,midway,rotate=40] {$\bm{\alpha}^{\ell}$};
  \draw [->] (-1.44,-.95) -- (-1.0,-0.6) node [below=-.1cm,midway,rotate=40] {$\bm{\beta}^{\ell}$};

  \draw [<-] (-.94,-.55) -- (-.5,-.9) node [above=-.1cm,midway,rotate=-40] {$\bm{\beta}^{\text{r}}$};
  \draw [<-] (-.56,-.95) -- (-0.975,-.625) node [below=-.1cm,midway,rotate=-40] {$\bm{\alpha}^{\text{r}}$};

\end{tikzpicture}
  \caption{SC-based decoding on a binary tree for $\mathcal{P}(8,4)$ and $\mathbf{v} = \{7,6,5,3,4,2,1,0\}$ ($\mathbf{s} = \{0,0,0,1,0,1,1,1\}$).}
  \label{fig:SCDec}
\end{figure}
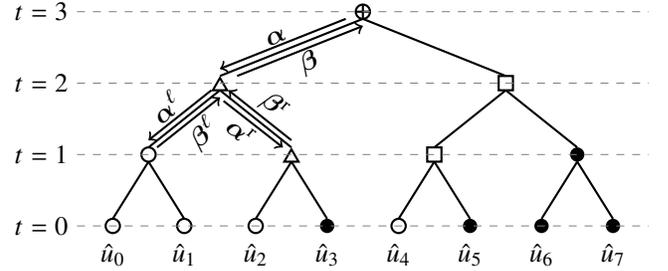

The $T=2^{t}$ elements of the left child node $\bm{\alpha}^\ell = \{\alpha^\ell_0,\alpha^\ell_1,\ldots,\alpha^\ell_{T-1}\}$ can be computed by the $F_t$ function, and those of the right child node $\bm{\alpha}^\text{r} = \{\alpha^\text{r}_0,\alpha^\text{r}_1,\ldots,\alpha^\text{r}_{T-1}\}$ can be computed by the $G_t$ function as
\begin{align}
\alpha^{\ell}_i =& F_t(\alpha_i,\alpha_{i+T}) \text{,} \label{eq:Ffunc1} \\
\alpha^{\text{r}}_i =& G_t(\alpha_i,\alpha_{i+T},\beta^\ell_i) \text{,} \label{eq:Gfunc1}
\end{align}
where
\begin{align}
F_t(a,b) =& 2\arctanh \left(\tanh\left(\frac{a}{2}\right)\tanh\left(\frac{b}{2}\right)\right) \text{,} \label{eq:Ffunc2} \\
\approx& \sgn(a)\sgn(b)\min(|a|,|b|) \text{,} \label{eq:Ffunc2HW} \\
G_t(a,b,c) =& b + \left(1-2c\right)a \text{.} \label{eq:Gfunc2}
\end{align}
Assume that the vector of relative reliabilities of bit-channels $\mathbf{v}$ is stored in memory and is available to the decoder. In SC and SCF decoding algorithms, when a leaf node is reached, the $i$-th bit $\hat{u}_i$ can be estimated as
\begin{equation}
\hat{u}_i =
  \begin{cases}
    0 \text{,} & \mbox{if } v_i \geq K \mbox{ or } \alpha_{i}\geq 0\text{,}\\
    1 \text{,} & \mbox{if } v_i < K \mbox{ and } \alpha_{i}< 0\text{,}
  \end{cases} \label{eq:leafSC}
\end{equation}
while in SCL decoding, at a leaf node we have
\begin{equation}
\hat{u}_i =
  \begin{cases}
    0 \text{,} & \mbox{if } v_i \geq K \text{,}\\
    0 \mbox{ and } 1 \text{,} & \mbox{if } v_i < K \text{.}
  \end{cases} \label{eq:leafSCL}
\end{equation}
As can be seen in (\ref{eq:leafSCL}), when an information bit is reached in SCL decoding, both of its possible values of $0$ and $1$ are considered. In order to limit the exponential growth in the complexity of the SCL decoder, at each bit estimation, only $L$ candidates are allowed to survive with the help of a path metric (PM) \cite{balatsoukas}. To this end, a sorter module is used to rank the PMs of the $2L$ generated candidates and selecting $L$ of them with the best PMs. After the estimation of bits by (\ref{eq:leafSC}) or (\ref{eq:leafSCL}), the left child and right child node messages $\bm{\beta}^\ell = \{\beta^\ell_0,\beta^\ell_1,\ldots,\beta^\ell_{T-1}\}$ and $\bm{\beta}^\text{r} = \{\beta^\text{r}_0,\beta^\text{r}_1,\ldots,\beta^\text{r}_{T-1}\}$ are used successively to calculate the $2T$ values of $\bm{\beta}$ as \cite{arikan}
\begin{equation}
\beta_i =
  \begin{cases}
    \beta^\ell_i\oplus \beta^\text{r}_i \text{,} & \text{if} \quad i < T \text{,}\\
    \beta^\text{r}_{i-T} \text{,} & \text{otherwise} \text{,}
  \end{cases}
  \label{eq:beta}
\end{equation}
where $\oplus$ is the bitwise XOR operation.

The depth-first binary tree search of SC-based decoding algorithms can be represented by a list of operations. Let $\mathbf{b}^i=\{b_{n-1}^i,b_{n-2}^i,\ldots,b_0^i\}$ represent the binary expansion of the integer $i$. The LLR value associated with $u_i$ can be calculated by a set of $F_t$ and $G_t$ operations as \cite{leroux}:
\begin{equation}
\begin{cases}
    F_t \text{,} & \mbox{if } b_t^{i}=0 \text{,}\\
    G_t \text{,} & \mbox{if } b_t^{i}=1 \text{.}
\end{cases} \label{eq:instFG}
\end{equation}
For example, the LLR value associated with $u_0$ in Fig.~\ref{fig:SCDec} can be calculated by performing $F_2$, $F_1$, and $F_0$, respectively, and the LLR value associated with $u_1$ in Fig.~\ref{fig:SCDec} can be calculated by performing $F_2$, $F_1$, and $G_0$, respectively. However, the calculation of the LLR value for $u_1$ can use the already calculated $F_2$ and $F_1$ operations in $u_0$. Let $M$ denote the minimum index in $\mathbf{b}^i$ such that $b_M^i = 1$. It is only required to perform $F_t$ or $G_t$ operations with $t\leq M$ because for $t>M$, the LLR values are already calculated for previous bits. For example, the list of operations associated with the SC-based decoder of Fig.~\ref{fig:SCDec} can be represented as $\{F_2,F_1,F_0,G_0,G_1,F_0,G_0,G_2,F_1,F_0,G_0,G_1,F_0,G_0\}$. It should be noted that since the hard estimate operations of (\ref{eq:leafSC}), (\ref{eq:leafSCL}), and (\ref{eq:beta}) are performed right after $F_t$ or $G_t$ functions at a leaf node and in the same time step, we do not include them in the list of operations. The list of operations for SC-based decoders can be generated directly on hardware by simple bitwise operations \cite{leroux,balatsoukas}.

It is worth mentioning that the list of operations for SC-based decoders is fixed for all rates and thus SC-based decoders are rate-flexible. However, the number of time steps required to finish the decoding process in SC-based decoders is at least $2N-2$\footnote{For SCL decoder, $K$ more time steps are needed to perform the PM computation and path pruning \cite{balatsoukas}. For SCF decoder, additional rounds of SC decoding add to the number of required time steps \cite{afisiadis}.}. This limits the latency and throughput of polar codes when decoded by SC-based decoders.

\subsection{Fast SC-Based Decoding} \label{sec:prel:FPDec}

In order to reduce the latency and increase the throughput of SC-based decoders for polar codes, special node structures are identified and the decoding is performed based on the LLR values at the intermediate levels in the SC-based decoding tree without the need of traversing it. It was shown in \cite{alamdar,sarkis} that four special nodes can be decoded efficiently in fast simplified SC (Fast-SSC) decoding without traversing the tree at the special nodes. Let $\mathbf{v}_t = \{v_{t_0},v_{t_1},\ldots,v_{t_{T-1}}\}$ represent a subset of $\mathbf{v}$ and $\mathbf{s}_t = \{s_{t_0},s_{t_1},\ldots,s_{t_{T-1}}\}$ represent a subset of $\mathbf{s}$ corresponding to a node of length $T$ in a polar code decoding tree. The four special nodes are:
\begin{itemize}
\item \emph{Rate-0 Node}: This node consists of only frozen bits, i.e., $v_{t_{i}} \geq K$ for any $i \in \{0,1,\ldots,T-1\}$ ($\mathbf{s}_t = \{0,0,\ldots,0\}$).
\item \emph{Rate-1 Node}: This node consists of only information bits, i.e., $v_{t_{i}} < K$ for any $i \in \{0,1,\ldots,T-1\}$ ($\mathbf{s}_t = \{1,1,\ldots,1\}$).
\item \emph{Repetition (Rep) Node}: This node consists of frozen bits except for the last bit which is an information bit, i.e., $v_{t_{T-1}} < K$ and $v_{t_{i}} \geq K$ for any $i \in \{0,1,\ldots,T-2\}$ ($\mathbf{s}_t = \{0,\ldots,0,0,1\}$).
\item \emph{Single parity-check (SPC) Node}: This node consists of information bits except for the first bit which is a frozen bit, i.e., $v_{t_{0}} \geq K$ and $v_{t_{i}} < K$ for any $i \in \{1,2,\ldots,T-1\}$ ($\mathbf{s}_t = \{0,1,1,\ldots,1\}$).
\end{itemize}
It was shown in \cite{hashemi_SSCL_TCASI,hashemi_FSSCL_TSP} that these nodes can be decoded efficiently also in simplified SCL (SSCL), SSCL-SPC, fast SSCL (Fast-SSCL), and Fast-SSCL-SPC decoding without the need for traversing the tree. This is performed by estimating bits one by one at an intermediate level of the decoding tree, thus generating only $2L$ candidates and selecting the best $L$ from them, similar to the conventional SCL decoding process. This guarantees that the sorter module which selects the $L$ candidates out of $2L$ remains the same as the conventional SCL decoder. The method was also applied to the SCF decoder which resulted in the Fast-SSCF decoder in \cite{giardFlip}. Recently, five new special nodes are observed in \cite{hanif} and efficient decoders that can be used in SC decoding were designed for them. These nodes are:
\begin{itemize}
\item \emph{Type-I Node}: This node consists of frozen bits except for the last two bits which are information bits, i.e., $v_{t_{T-1}} < K$, $v_{t_{T-2}} < K$, and $v_{t_{i}} \geq K$ for any $i \in \{0,1,\ldots,T-3\}$ ($\mathbf{s}_t = \{0,\ldots,0,1,1\}$).
\item \emph{Type-II Node}: This node consists of frozen bits except for the last three bits which are information bits, i.e., $v_{t_{T-1}} < K$, $v_{t_{T-2}} < K$, $v_{t_{T-3}} < K$, and $v_{t_{i}} \geq K$ for any $i \in \{0,1,\ldots,T-4\}$ ($\mathbf{s}_t = \{0,\ldots,0,1,1,1\}$).
\item \emph{Type-III Node}: This node consists of information bits except for the first two bits which are frozen bits, i.e., $v_{t_{0}} \geq K$, $v_{t_{1}} \geq K$, and $v_{t_{i}} < K$ for any $i \in \{2,3,\ldots,T-1\}$ ($\mathbf{s}_t = \{0,0,1,\ldots,1\}$).
\item \emph{Type-IV Node}: This node consists of information bits except for the first three bits which are frozen bits, i.e., $v_{t_{0}} \geq K$, $v_{t_{1}} \geq K$, $v_{t_{2}} \geq K$, and $v_{t_{i}} < K$ for any $i \in \{3,4,\ldots,T-1\}$ ($\mathbf{s}_t = \{0,0,0,1,\ldots,1\}$).
\item \emph{Type-V Node}: This node consists of frozen bits except for the bits $T-5$, $T-3$, $T-2$, and $T-1$ which are information bits, i.e., $v_{t_{T-1}} < K$, $v_{t_{T-2}} < K$, $v_{t_{T-3}} < K$, $v_{t_{T-4}} \geq K$, $v_{t_{T-5}} < K$, and $v_{t_{i}} \geq K$ for any $i \in \{0,1,\ldots,T-6\}$ ($\mathbf{s}_t = \{0,\ldots,0,1,0,1,1,1\}$).
\end{itemize}
It was shown in \cite{hanif_SCL} that these new nodes can be decoded efficiently to improve the speed of SCL decoding. However, the drawback of using these new nodes when implementing the decoder on hardware is that these nodes are based on multiple bit estimations at a time, thus producing more than $2L$ candidates in each decoding step. Therefore, a large sorter is required to select the final $L$ caldidates which adversely affects the hardware implementation complexity. In particular, at each decoding step, Type-I node produces $4L$ candidates to account for all the cases for its two information bits, Type-II node produces $8L$ candidates to account for all the cases for its three information bits, and Type-V node produces $16L$ candidates to account for all the cases for its four information bits. Moreover, Type-III node is decoded using two parallel SPC node decoders, and Type-IV node starts by decoding a Rep node of length four followed by four parallel SPC node decoders \cite{hanif_SCL}.

\begin{figure}
  \centering
  \begin{tikzpicture}[scale=1.9, thick]
\newcommand\Triangle[1]{-- ++(0:2*#1) -- ++(120:2*#1) --cycle}
\newcommand\Square[1]{+(-#1,-#1) rectangle +(#1,#1)}

  \draw (0,0) circle [radius=.05];
  
  \draw (-.05,0) -- (.05,0);
  \draw (0,-.05) -- (0,.05);

  \draw (-1.05,-.55) \Triangle{.05};
  \draw (1,-.5) \Square{.05};

  \draw (0,-.05) -- (-1,-.45);
  \draw (0,-.05) -- (1,-.45);

  \draw [very thin,gray,dashed] (-2,0) -- (2,0);
  \draw [very thin,gray,dashed] (-2,-.5) -- (2,-.5);

  \node at (-2.25,0) {$t=3$};
  \node at (-2.25,-.5) {$t=2$};

  \node at (0,.25) {Type-V};
  \node at (-1,-.75) {Rep};
  \node at (1,-.75) {SPC};

  \draw [->] (-.12,-.05) -- (-1,-.4) node [above=-.1cm,midway,rotate=25] {$\bm{\alpha}$};
  \draw [->] (-.88,-.45) -- (0,-.1) node [below=-.1cm,midway,rotate=25] {$\bm{\beta}$};

\end{tikzpicture}
  \caption{Fast SC-based decoding on a binary tree for $\mathcal{P}(8,4)$ and $\mathbf{v} = \{7,6,5,3,4,2,1,0\}$ ($\mathbf{s} = \{0,0,0,1,0,1,1,1\}$).}
  \label{fig:Fast-SSCDec}
\end{figure}
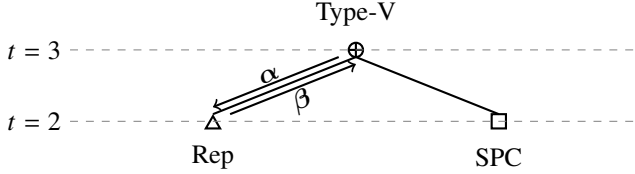

The pruned decoding tree for the same example as in Fig.~\ref{fig:SCDec} is shown in Fig.~\ref{fig:Fast-SSCDec}. If the new nodes are not taken into account, $\mathcal{P}(8,4)$ can be decoded in four time steps by traversing the tree for one level and decode the resulting Rep and SPC nodes. The resulting list of operations for the decoder would be $\{F_2,\text{Rep}_2,G_2,\text{SPC}_2\}$, where $\text{Rep}_t$ and $\text{SPC}_t$ represent the decoding of Rep and SPC nodes of length $T=2^t$, respectively. However, by considering the new nodes, the decoder can immediately decode the received vector by decoding the Type-V node. The corresponding list of operations would be $\{\text{Type-V}_3\}$, where $\text{Type-V}_t$ represents the decoding of Type-V nodes of length $T=2^t$. The operations which are performed in fast SC-based decoders are summarized in Table~\ref{tab:Oper}. Note that $F_t$ and $G_t$ operations are common between conventional SC-based and fast SC-based decoding algorithms. In the hardware implementation of fast SC-based decoders, this list of operations is stored in memory and is fed into the decoder to perform decoding \cite{sarkis,hashemi_SSCL_TCASI,hashemi_FSSCL_TSP}.

\begin{table}[t!]
\centering
\caption{Different operations that are supported in SC-based decoding algorithms.}
\label{tab:Oper}
\setlength{\extrarowheight}{2.5pt}
\begin{tabular}{lll}
\toprule

Operation & Description & Decoder \\

\midrule

$F_t$ & Calculate $\bm{\alpha}^\ell$ at level $t$. & SC-based \\
$G_t$ & Calculate $\bm{\alpha}^\text{r}$ at level $t$. & SC-based \\
$\text{Rate-0}_t$ & Decode Rate-0 node of length $2^t$. & Fast SC-based \\
$\text{Rate-1}_t$ & Decode Rate-1 node of length $2^t$. & Fast SC-based \\
$\text{Rep}_t$ & Decode Rep node of length $2^t$. & Fast SC-based \\
$\text{SPC}_t$ & Decode SPC node of length $2^t$. & Fast SC-based \\
$\text{Type-I}_t$ & Decode Type-I node of length $2^t$. & Fast SC-based \\
$\text{Type-II}_t$ & Decode Type-II node of length $2^t$. & Fast SC-based \\
$\text{Type-III}_t$ & Decode Type-III node of length $2^t$. & Fast SC-based \\
$\text{Type-IV}_t$ & Decode Type-IV node of length $2^t$. & Fast SC-based \\
$\text{Type-V}_t$ & Decode Type-V node of length $2^t$. & Fast SC-based \\

\bottomrule
\end{tabular}
\end{table}

Let us consider the example in Fig.~\ref{fig:Fast-SSCDec}. If the rate of the code changes from $1/2$ to $5/8$, the list of operations also changes as shown in Fig.~\ref{fig:Fast-SSCDec58}. Without using the new nodes, the list of operations becomes $\{F_2,\text{Rep}_2,G_2,\text{Rate-1}_2\}$, and by considering the new nodes it becomes $\{\text{Type-IV}_3\}$. Therefore, as the rate changes, the list of operations changes. The resulting decoder is therefore not rate-flexible. For applications that support codes with multiple rates, for each rate, the list of operations has to be stored in memory to make the decoder flexible. However, this results in high memory usage when implemented on hardware.

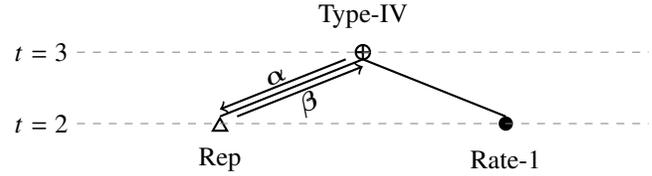
\begin{figure}
  \centering
  \begin{tikzpicture}[scale=1.9, thick]
\newcommand\Triangle[1]{-- ++(0:2*#1) -- ++(120:2*#1) --cycle}
\newcommand\Square[1]{+(-#1,-#1) rectangle +(#1,#1)}

  \draw (0,0) circle [radius=.05];
  
  \draw (-.05,0) -- (.05,0);
  \draw (0,-.05) -- (0,.05);

  \draw (-1.05,-.55) \Triangle{.05};
  \fill (1,-.5) circle [radius=.05];

  \draw (0,-.05) -- (-1,-.45);
  \draw (0,-.05) -- (1,-.45);

  \draw [very thin,gray,dashed] (-2,0) -- (2,0);
  \draw [very thin,gray,dashed] (-2,-.5) -- (2,-.5);

  \node at (-2.25,0) {$t=3$};
  \node at (-2.25,-.5) {$t=2$};

  \node at (0,.25) {Type-IV};
  \node at (-1,-.75) {Rep};
  \node at (1,-.75) {Rate-1};

  \draw [->] (-.12,-.05) -- (-1,-.4) node [above=-.1cm,midway,rotate=25] {$\bm{\alpha}$};
  \draw [->] (-.88,-.45) -- (0,-.1) node [below=-.1cm,midway,rotate=25] {$\bm{\beta}$};

\end{tikzpicture}
  \caption{Fast SC-based decoding on a binary tree for $\mathcal{P}(8,5)$ and $\mathbf{v} = \{7,6,5,3,4,2,1,0\}$ ($\mathbf{s} = \{0,0,0,1,1,1,1,1\}$).}
  \label{fig:Fast-SSCDec58}
\end{figure}

\section{Rate-Flexible Fast Polar Decoding} \label{sec:RFFPD}

The high memory usage of storing the list of operations can be mitigated by generating the list of operations on hardware as the decoding proceeds. A rudimentary approach would be to generate the vector $\mathbf{s}_t$ from $K$ and the vector $\mathbf{v}_t$ using comparators, and check the pattern of information and frozen bits in $\mathbf{s}_t$ for every encountered node. This is shown in Fig.~\ref{fig:logic-nodes} for determining Rate-0, Rate-1, Rep, and SPC nodes of length $8$. It should be noted that the comparators in Fig.~\ref{fig:logic-nodes:a} have two inputs $A$ and $B$, and an output $C$ where
\begin{equation}
C =
  \begin{cases}
    0 \text{,} & \text{if } A \geq B \text{,}\\
    1 \text{,} & \text{if } A < B \text{.}
  \end{cases}
  \label{eq:comparator}
\end{equation}
The problem with this approach is that for nodes of large length, there is a high hardware complexity overhead in generating $\mathbf{s}_t$ from $K$ and $\mathbf{v}_t$, and determining the node types. Moreover, the module that generates the list of operations should account for the largest possible node which is the root node in the decoding tree with size $N$. This results in a large critical path which limits the operating frequency.

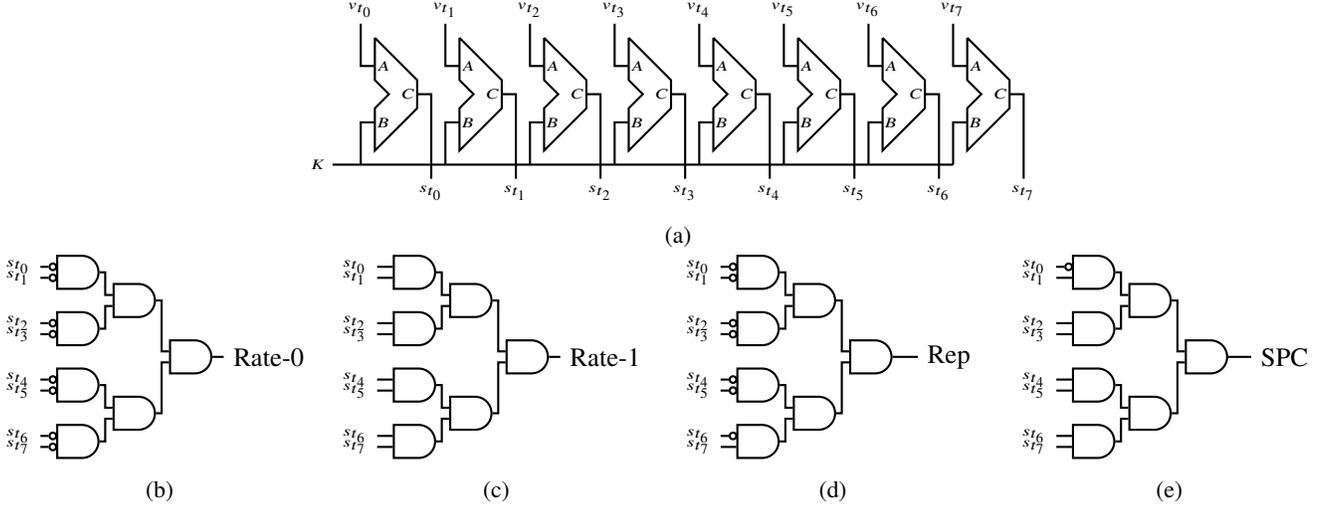
\begin{figure*}
\begin{subfigure}[]{\textwidth}
  \centering
  \begin{tikzpicture}[circuit logic US, tiny circuit symbols, scale=.75, thick]

\foreach \x in {0,1.5,...,10.5}
{
\draw (\x,0) -- ++(.75,-.75) -- ++(0,-.5) -- ++(-.75,-.75) -- ++(0,.75) -- ++(.25,.25) -- ++(-.25,.25) -- ++(0,.75);

\node (A) at (\x+.15,-.5) {\tiny $A$};
\node (B) at (\x+.15,-1.5) {\tiny $B$};
\node (C) at (\x+.6,-1) {\tiny $C$};

\draw (\x,-.5) -- ++(-.25,0) -- ++(0,.75);
\draw (\x,-1.5) -- ++(-.25,0) -- ++(0,-.75);
\draw (\x+.75,-1) -- ++(.25,0) -- ++(0,-1.5);
}

\draw (-.75,-2.25) -- (10.25,-2.25);

\node at (-1,-2.25) {\tiny $K$};

\node at (-.25,.5) {\tiny $v_{t_0}$};
\node at (1.25,.5) {\tiny $v_{t_1}$};
\node at (2.75,.5) {\tiny $v_{t_2}$};
\node at (4.25,.5) {\tiny $v_{t_3}$};
\node at (5.75,.5) {\tiny $v_{t_4}$};
\node at (7.25,.5) {\tiny $v_{t_5}$};
\node at (8.75,.5) {\tiny $v_{t_6}$};
\node at (10.25,.5) {\tiny $v_{t_7}$};

\node at (1,-2.75) {\tiny $s_{t_0}$};
\node at (2.5,-2.75) {\tiny $s_{t_1}$};
\node at (4,-2.75) {\tiny $s_{t_2}$};
\node at (5.5,-2.75) {\tiny $s_{t_3}$};
\node at (7,-2.75) {\tiny $s_{t_4}$};
\node at (8.5,-2.75) {\tiny $s_{t_5}$};
\node at (10,-2.75) {\tiny $s_{t_6}$};
\node at (11.5,-2.75) {\tiny $s_{t_7}$};

\end{tikzpicture}
  \caption{}
  \label{fig:logic-nodes:a}
\end{subfigure}
\begin{subfigure}[]{.24\textwidth}
  \centering
  \begin{tikzpicture}[circuit logic US, tiny circuit symbols, scale=1, thick]
    \node [and gate, inputs=ii] (and1) at (0,0) {};
    \node [and gate, inputs=ii] (and2) at (0,-.75) {};
    \node [and gate, inputs=ii] (and3) at (0,-1.5) {};
    \node [and gate, inputs=ii] (and4) at (0,-2.25) {};
    \node [and gate, inputs=nn] (and5) at (.75,-.375) {};
    \node [and gate, inputs=nn] (and6) at (.75,-1.875) {};
    \node [and gate, inputs=nn] (and7) at (1.5,-1.125) {};

\node (s0) at ([xshift=-4mm]and1.input 1) {\tiny $s_{t_0}$};
\node (s1) at ([xshift=-4mm]and1.input 2) {\tiny $s_{t_1}$};
\node (s2) at ([xshift=-4mm]and2.input 1) {\tiny $s_{t_2}$};
\node (s3) at ([xshift=-4mm]and2.input 2) {\tiny $s_{t_3}$};
\node (s4) at ([xshift=-4mm]and3.input 1) {\tiny $s_{t_4}$};
\node (s5) at ([xshift=-4mm]and3.input 2) {\tiny $s_{t_5}$};
\node (s6) at ([xshift=-4mm]and4.input 1) {\tiny $s_{t_6}$};
\node (s7) at ([xshift=-4mm]and4.input 2) {\tiny $s_{t_7}$};
\draw (and1.input 1) -- (s0);
\draw (and1.input 2) -- (s1);
\draw (and2.input 1) -- (s2);
\draw (and2.input 2) -- (s3);
\draw (and3.input 1) -- (s4);
\draw (and3.input 2) -- (s5);
\draw (and4.input 1) -- (s6);
\draw (and4.input 2) -- (s7);
\draw (and1.output) --++(0:.75mm) |- (and5.input 1);
\draw (and2.output) --++(0:.75mm) |- (and5.input 2);
\draw (and3.output) --++(0:.75mm) |- (and6.input 1);
\draw (and4.output) --++(0:.75mm) |- (and6.input 2);
\draw (and5.output) --++(0:.75mm) |- (and7.input 1);
\draw (and6.output) --++(0:.75mm) |- (and7.input 2);
\node (out) at ([xshift=7.5mm]and7.output) {Rate-0};
\draw (and7.output) -- (out);
\end{tikzpicture}
  \caption{}
\end{subfigure}
\begin{subfigure}[]{.24\textwidth}
  \centering
  \begin{tikzpicture}[circuit logic US, tiny circuit symbols, scale=1, thick]
    \node [and gate, inputs=nn] (and1) at (0,0) {};
    \node [and gate, inputs=nn] (and2) at (0,-.75) {};
    \node [and gate, inputs=nn] (and3) at (0,-1.5) {};
    \node [and gate, inputs=nn] (and4) at (0,-2.25) {};
    \node [and gate, inputs=nn] (and5) at (.75,-.375) {};
    \node [and gate, inputs=nn] (and6) at (.75,-1.875) {};
    \node [and gate, inputs=nn] (and7) at (1.5,-1.125) {};

\node (s0) at ([xshift=-5mm]and1.input 1) {\tiny $s_{t_0}$};
\node (s1) at ([xshift=-5mm]and1.input 2) {\tiny $s_{t_1}$};
\node (s2) at ([xshift=-5mm]and2.input 1) {\tiny $s_{t_2}$};
\node (s3) at ([xshift=-5mm]and2.input 2) {\tiny $s_{t_3}$};
\node (s4) at ([xshift=-5mm]and3.input 1) {\tiny $s_{t_4}$};
\node (s5) at ([xshift=-5mm]and3.input 2) {\tiny $s_{t_5}$};
\node (s6) at ([xshift=-5mm]and4.input 1) {\tiny $s_{t_6}$};
\node (s7) at ([xshift=-5mm]and4.input 2) {\tiny $s_{t_7}$};
\draw (and1.input 1) -- (s0);
\draw (and1.input 2) -- (s1);
\draw (and2.input 1) -- (s2);
\draw (and2.input 2) -- (s3);
\draw (and3.input 1) -- (s4);
\draw (and3.input 2) -- (s5);
\draw (and4.input 1) -- (s6);
\draw (and4.input 2) -- (s7);
\draw (and1.output) --++(0:.75mm) |- (and5.input 1);
\draw (and2.output) --++(0:.75mm) |- (and5.input 2);
\draw (and3.output) --++(0:.75mm) |- (and6.input 1);
\draw (and4.output) --++(0:.75mm) |- (and6.input 2);
\draw (and5.output) --++(0:.75mm) |- (and7.input 1);
\draw (and6.output) --++(0:.75mm) |- (and7.input 2);
\node (out) at ([xshift=7.5mm]and7.output) {Rate-1};
\draw (and7.output) -- (out);
\end{tikzpicture}
  \caption{}
\end{subfigure}
\begin{subfigure}[]{.24\textwidth}
  \centering
  \begin{tikzpicture}[circuit logic US, tiny circuit symbols, scale=1, thick]
    \node [and gate, inputs=ii] (and1) at (0,0) {};
    \node [and gate, inputs=ii] (and2) at (0,-.75) {};
    \node [and gate, inputs=ii] (and3) at (0,-1.5) {};
    \node [and gate, inputs=in] (and4) at (0,-2.25) {};
    \node [and gate, inputs=nn] (and5) at (.75,-.375) {};
    \node [and gate, inputs=nn] (and6) at (.75,-1.875) {};
    \node [and gate, inputs=nn] (and7) at (1.5,-1.125) {};

\node (s0) at ([xshift=-4mm]and1.input 1) {\tiny $s_{t_0}$};
\node (s1) at ([xshift=-4mm]and1.input 2) {\tiny $s_{t_1}$};
\node (s2) at ([xshift=-4mm]and2.input 1) {\tiny $s_{t_2}$};
\node (s3) at ([xshift=-4mm]and2.input 2) {\tiny $s_{t_3}$};
\node (s4) at ([xshift=-4mm]and3.input 1) {\tiny $s_{t_4}$};
\node (s5) at ([xshift=-4mm]and3.input 2) {\tiny $s_{t_5}$};
\node (s6) at ([xshift=-4mm]and4.input 1) {\tiny $s_{t_6}$};
\node (s7) at ([xshift=-5mm]and4.input 2) {\tiny $s_{t_7}$};
\draw (and1.input 1) -- (s0);
\draw (and1.input 2) -- (s1);
\draw (and2.input 1) -- (s2);
\draw (and2.input 2) -- (s3);
\draw (and3.input 1) -- (s4);
\draw (and3.input 2) -- (s5);
\draw (and4.input 1) -- (s6);
\draw (and4.input 2) -- (s7);
\draw (and1.output) --++(0:.75mm) |- (and5.input 1);
\draw (and2.output) --++(0:.75mm) |- (and5.input 2);
\draw (and3.output) --++(0:.75mm) |- (and6.input 1);
\draw (and4.output) --++(0:.75mm) |- (and6.input 2);
\draw (and5.output) --++(0:.75mm) |- (and7.input 1);
\draw (and6.output) --++(0:.75mm) |- (and7.input 2);
\node (out) at ([xshift=7.5mm]and7.output) {Rep};
\draw (and7.output) -- (out);
\end{tikzpicture}
  \caption{}
\end{subfigure}
\begin{subfigure}[]{.24\textwidth}
  \centering
  \begin{tikzpicture}[circuit logic US, tiny circuit symbols, scale=1, thick]
    \node [and gate, inputs=in] (and1) at (0,0) {};
    \node [and gate, inputs=nn] (and2) at (0,-.75) {};
    \node [and gate, inputs=nn] (and3) at (0,-1.5) {};
    \node [and gate, inputs=nn] (and4) at (0,-2.25) {};
    \node [and gate, inputs=nn] (and5) at (.75,-.375) {};
    \node [and gate, inputs=nn] (and6) at (.75,-1.875) {};
    \node [and gate, inputs=nn] (and7) at (1.5,-1.125) {};

\node (s0) at ([xshift=-4mm]and1.input 1) {\tiny $s_{t_0}$};
\node (s1) at ([xshift=-5mm]and1.input 2) {\tiny $s_{t_1}$};
\node (s2) at ([xshift=-5mm]and2.input 1) {\tiny $s_{t_2}$};
\node (s3) at ([xshift=-5mm]and2.input 2) {\tiny $s_{t_3}$};
\node (s4) at ([xshift=-5mm]and3.input 1) {\tiny $s_{t_4}$};
\node (s5) at ([xshift=-5mm]and3.input 2) {\tiny $s_{t_5}$};
\node (s6) at ([xshift=-5mm]and4.input 1) {\tiny $s_{t_6}$};
\node (s7) at ([xshift=-5mm]and4.input 2) {\tiny $s_{t_7}$};
\draw (and1.input 1) -- (s0);
\draw (and1.input 2) -- (s1);
\draw (and2.input 1) -- (s2);
\draw (and2.input 2) -- (s3);
\draw (and3.input 1) -- (s4);
\draw (and3.input 2) -- (s5);
\draw (and4.input 1) -- (s6);
\draw (and4.input 2) -- (s7);
\draw (and1.output) --++(0:.75mm) |- (and5.input 1);
\draw (and2.output) --++(0:.75mm) |- (and5.input 2);
\draw (and3.output) --++(0:.75mm) |- (and6.input 1);
\draw (and4.output) --++(0:.75mm) |- (and6.input 2);
\draw (and5.output) --++(0:.75mm) |- (and7.input 1);
\draw (and6.output) --++(0:.75mm) |- (and7.input 2);
\node (out) at ([xshift=7.5mm]and7.output) {SPC};
\draw (and7.output) -- (out);
\end{tikzpicture}
  \caption{}
\end{subfigure}
  \caption{Determination of node types for fast SC-based decoding in a node of length $T=8$. (a) generation of $\mathbf{s}_t$ from $K$ and $\mathbf{v}_t$, (b) Rate-0 node, (c) Rate-1 node, (d) Rep node, (e) SPC node.}
  \label{fig:logic-nodes}
\end{figure*}

In order to tackle the above issue, the idea is to exploit the inherent order in the Bhattacharyya parameters of the bit-channels. Let $W_i$ and $W_j$ be the bit-channels corresponding to $u_i$ and $u_j$, and let $\mathbf{b}^i$ and $\mathbf{b}^j$ be the binary expansions of the integers $i$ and $j$. In \cite{Sc16,BDOT16} a partial order between the polarized bit-channels was introduced. In particular, it was proven that $W_i$ is stochastically degraded with respect to $W_j$, i.e., $W_i \prec W_j$, when one of the following two properties hold:

\begin{itemize}
\item \emph{Addition Property} \cite{mondelliUPO}: There exists $k\in \{0, 1, \ldots, n-1\}$ such that
\begin{equation}\label{eq:add}
  \begin{cases}
    b_m^i = b_m^j \text{,} & \text{if} \quad m \neq k \text{,}\\
    b_k^i = 0 \text{,}\\
    b_k^j = 1 \text{.}\\
  \end{cases}
\end{equation}
\item \emph{Left-Swap Property}\cite{mondelliUPO}: There exist $k, l\in \{0, 1, \ldots, n-1\}$ such that $l < k$ and
\begin{equation}\label{eq:leftswap}
  \begin{cases}
    b_m^i = b_m^j \text{,} & \text{if} \quad m \neq k \mbox{, } m \neq l \text{,}\\
    b_k^i = b_l^j = 0 \text{,}\\
    b_l^i = b_k^j = 1 \text{.}\\
  \end{cases}
\end{equation}
\end{itemize}
Recall that, if $W_i \prec W_j$, then all the reliability measures of $W_i$ are worse than those of $W_j$, i.e., $W_i$ has smaller mutual information, larger Bhattacharyya parameter, and larger error probability. Consequently, if $u_j$ belongs to the frozen set, then also $u_i$ belongs to the frozen set. Furthermore, if $u_i$ belongs to the information set, then also $u_j$ belongs to the information set. By using the two properties above, it was shown in \cite{mondelliUPO} that it suffices to compute the reliability of a sublinear fraction of channels in order to identify the frozen and the information sets.

Another option to find an ordering between the Bhattacharyya parameters of the bit-channels can be described as follows. Consider the transmission over a BMS channel $W$ with Bhattacharyya parameter $Z(W)$ and define the synthetic channels $W^0$ and $W^1$ as 
\begin{equation}
\begin{split}
W^0(y_1, y_2\mid x_1) & = \sum_{x_2} \frac{1}{2}W(y_1\mid x_1 \oplus x_2) W(y_2\mid x_2),\\
W^1(y_1, y_2, x_1\mid x_2) & = \frac{1}{2}W(y_1\mid x_1 \oplus x_2) W(y_2\mid x_2).\\
\end{split}
\end{equation}
Then, the following inequalities between $Z(W^{0})$, $Z(W^{1})$ and $Z(W)$ hold
\begin{equation}\label{eq:plusB}
\begin{split}
Z(W)\sqrt{2-Z(W)^2}&\le Z(W^0)\le 2Z(W)-Z(W)^2,\\
&Z(W^1)=Z(W)^2,
\end{split}
\end{equation}
which follow from Proposition 5 of \cite{arikan} and from Exercise 4.62 of \cite{RiU08}. 
Furthermore, the bit-channel $W_i$ corresponding to $u_i$ is given by the recursive formula below:
\begin{equation}
W_i = (((W^{b_{n-1}^i})^{b_{n-2}^i})^{\ldots})^{b_{0}^i}.
\end{equation}
In what follows, we will denote by $Z_i$ the Bhattacharyya parameter of $W_i$.

At this point, we are ready to state and prove the first result of this paper, which concerns the identification of Rate-0, Rate-1, Rep, and SPC nodes. 
\begin{theorem}\label{th:th1}
Consider a node of length $T=2^t$ in a polar code of length $N=2^n$. Then, the following properties hold:
\begin{enumerate}
\item If $v_{t_{T-1}} \geq K$, i.e., $s_{t_{T-1}} = 0$, then the node represents a Rate-0 node.
\item If $v_{t_0} < K$, i.e., $s_{t_0} = 1$, then the node represents a Rate-1 node.
\item If $v_{t_{T-1}} < K$ and $v_{t_{T-2}} \geq K$, i.e., $s_{t_{T-1}} = 1$ and $s_{t_{T-2}} = 0$, then the node represents a Rep node.
\item If $v_{t_0} \geq K$ and $v_{t_1} < K$, i.e., $s_{t_0} = 0$ and $s_{t_1} = 1$, then the node represents an SPC node.
\end{enumerate} 
\end{theorem}
\begin{proof}
\begin{enumerate}
\item Note that $\mathbf{b}^{T-1}=\{1, \ldots, 1\}$. By using the addition property \eqref{eq:add}, we obtain that $W_{i}\prec W_{T-1}$ for any $i\in\{0, 1, \ldots, T-2\}$. Hence, as $v_{t_{T-1}} \geq K$, $v_{t_i} \geq K$ for any $i\in\{0, 1, \ldots, T-2\}$. This means that the polar code consists of only frozen bits, i.e., it is a Rate-0 node.

\item Note that $\mathbf{b}^{0}=\{0, \ldots, 0\}$. By using the addition property \eqref{eq:add}, we obtain that $W_{0}\prec W_{i}$ for any $i\in\{1, 2, \ldots, T-1\}$. Hence, as $v_{t_0} < K$, $v_{t_i} < K$ for any $i\in\{1, 2, \ldots, T-1\}$. This means that the polar code consists of only information bits, i.e., it is a Rate-1 node.

\item Note that $\mathbf{b}^{T-2}=\{1, \ldots, 1, 0\}$. By using the addition property \eqref{eq:add} and the left-swap property \eqref{eq:leftswap}, we obtain that $W_{i}\prec W_{T-2}$ for any $i\in\{0, 1, \ldots, T-3\}$. Hence, as $v_{t_{T-2}} \geq K$, $v_{t_i} \geq K$ for any $i\in\{0, 1, \ldots, T-3\}$. As $v_{t_{T-1}} < K$, the polar code consists of frozen bits except for the last bit which is an information bit, i.e., it is a Rep node.

\item Note that $\mathbf{b}^{1}=\{0, \ldots, 0, 1\}$. By using the addition property \eqref{eq:add} and the left-swap property \eqref{eq:leftswap}, we obtain that $W_{1}\prec W_{i}$ for any $i\in\{2, 3, \ldots, T-1\}$. Hence, as $v_{t_1} < K$, $v_{t_i} < K$ for any $i\in\{2, 3, \ldots, T-1\}$. As $v_{t_0} \geq K$, the polar code consists of information bits except for the first bit which is a frozen bit, i.e., it is an SPC node.
\end{enumerate}
\end{proof}

In the proof of Theorem~\ref{th:th1}, we used the fact that for any node of length $T=2^t$ in a polar code of length $N=2^n$, the $n$-bit binary expansions of the integers corresponding to the bit-channels in the node are equal in the bits $\{n-1,n-2,\ldots,t\}$, and are different in the bits $\{t-1,t-2,\ldots,0\}$. An immediate consequence of Theorem~\ref{th:th1} is that, by checking only one value, we can find out if a constituent node is either a Rate-0 or a Rate-1 node. Furthermore, by checking only two values, we can find out if a constituent node is either a Rep or an SPC node. This observation significantly reduces the hardware complexity associated with the on-line node identification. In addition, the proposed approach is independent of the node length, making it suitable for codes of any length and rate. Fig.~\ref{fig:logicall} shows the circuit required to generate the list of operations on-line for any node of length $T$. It can be seen that the circuit consists of only four comparators, three NOT gates, and two AND gates.

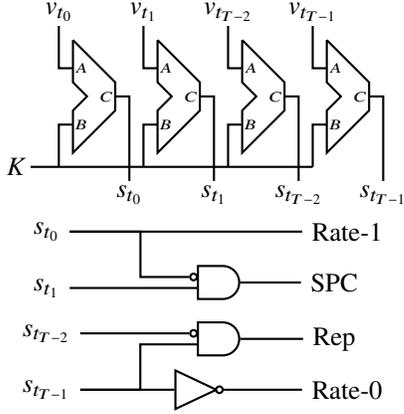
\begin{figure}
  \centering
  \begin{tikzpicture}[circuit logic US, tiny circuit symbols, scale=.75, thick]

\foreach \x in {0,1.5,...,4.5}
{
\draw (\x,0) -- ++(.75,-.75) -- ++(0,-.5) -- ++(-.75,-.75) -- ++(0,.75) -- ++(.25,.25) -- ++(-.25,.25) -- ++(0,.75);

\node (A) at (\x+.15,-.5) {\tiny $A$};
\node (B) at (\x+.15,-1.5) {\tiny $B$};
\node (C) at (\x+.6,-1) {\tiny $C$};

\draw (\x,-.5) -- ++(-.25,0) -- ++(0,.75);
\draw (\x,-1.5) -- ++(-.25,0) -- ++(0,-.75);
\draw (\x+.75,-1) -- ++(.25,0) -- ++(0,-1.5);
}

\draw (-.75,-2.25) -- (4.25,-2.25);

\node at (-1,-2.25) {$K$};

\node at (-.25,.5) {$v_{t_0}$};
\node at (1.25,.5) {$v_{t_1}$};
\node at (2.75,.5) {$v_{t_{T-2}}$};
\node at (4.25,.5) {$v_{t_{T-1}}$};

\node at (1,-2.75) {$s_{t_0}$};
\node at (2.5,-2.75) {$s_{t_1}$};
\node at (4,-2.75) {$s_{t_{T-2}}$};
\node at (5.5,-2.75) {$s_{t_{T-1}}$};

\end{tikzpicture}
  \begin{tikzpicture}[circuit logic US, tiny circuit symbols, scale=1, thick]
    \node [and gate, inputs=in] (and1) at (0,0) {};
    \node [and gate, inputs=in] (and2) at (0,-.75) {};
    \node [not gate] (not1) at ([yshift=-.75cm]and2.input 1) {};

\node (s0) at ([xshift=-2cm,yshift=.75cm]and1.input 2) {$s_{t_0}$};
\node (s1) at ([xshift=-2cm]and1.input 2) {$s_{t_1}$};
\node (s6) at ([xshift=-1.9cm]and2.input 1) {$s_{t_{T-2}}$};
\node (s7) at ([xshift=-1.9cm,yshift=-.75cm]and2.input 1) {$s_{t_{T-1}}$};
\draw (s0) --++(0:1.25cm) |- (and1.input 1);
\draw (s1) -- (and1.input 2);
\draw (s6) -- (and2.input 1);
\draw (s7) --++(0:1.25cm) |- (and2.input 2);
\draw (s7) -- (not1.input);

\node (out-rate1) at ([xshift=4cm]s0) {Rate-1};
\node (out-spc) at ([xshift=1.25cm]and1.output) {SPC};
\node (out-rep) at ([xshift=1.25cm]and2.output) {Rep};
\node (out-rate0) at ([xshift=4cm]s7) {Rate-0};

\draw (s0) -- (out-rate1);
\draw (and1.output) -- (out-spc);
\draw (and2.output) -- (out-rep);
\draw (not1.output) -- (out-rate0);
\end{tikzpicture}
  \caption{Efficient generation of the list of operations on hardware.}
  \label{fig:logicall}
\end{figure}

Let us now state and prove the second result of this paper, which concerns the identification of Type-I, Type-II, Type-III, Type-IV, and Type-V nodes.
\begin{theorem}\label{th:th2}
Consider a node of length $T=2^t$ in a polar code of length $N=2^n$. Then, the following properties hold:
\begin{enumerate}
\item If $v_{t_{T-1}} < K$, $v_{t_{T-2}} < K$, and $v_{t_{T-3}} \geq K$, then the node represents a Type-I node.
\item If $v_{t_{T-1}} < K$, $v_{t_{T-2}} < K$, $v_{t_{T-3}} < K$, and $v_{t_{T-5}} \geq K$, then the node represents a Type-II node.
\item If $v_{t_0} \geq K$, $v_{t_1} \geq K$, and $v_{t_2} < K$, then the node represents a Type-III node.
\item If $v_{t_0} \geq K$, $v_{t_1} \geq K$, $v_{t_2} \geq K$, and $v_{t_4} < K$, then the node represents a Type-IV node.
\item If $v_{t_{T-1}} < K$, $v_{t_{T-2}} < K$, $v_{t_{T-3}} < K$, $v_{t_{T-4}} \geq K$, $v_{t_{T-5}} < K$, and $v_{t_{T-9}} \geq K$, then the node represents a Type-V node.
\end{enumerate} 
\end{theorem}
\begin{proof}
\begin{enumerate}
\item Note that $\mathbf{b}^{T-3}=\{1, \ldots, 1, 0, 1\}$. By using the addition property \eqref{eq:add} and the left-swap property \eqref{eq:leftswap}, we obtain that $W_{i}\prec W_{T-3}$ for any $i\in\{0, 1, \ldots, T-4\}$. Hence, as $v_{t_{T-3}} \geq K$, $v_{t_i} \geq K$ for any $i\in\{0, 1, \ldots, T-4\}$. As $v_{t_{T-1}} < K$ and $v_{t_{T-2}} < K$, the node consists of frozen bits except for the last two bits which are information bits, i.e., it is a Type-I node.

\item Note that $\mathbf{b}^{T-5}=\{1, \ldots, 1, 0, 1, 1\}$. By using the addition property \eqref{eq:add} and the left-swap property \eqref{eq:leftswap}, we obtain that $W_{i}\prec W_{T-5}$ for any $i\in\{0, 1, \ldots, T-6\}$. Hence, as $v_{t_{T-5}} \geq K$, $v_{t_i} \geq K$ for any $i\in\{0, 1, \ldots, T-6\}$. Furthermore, note that $\mathbf{b}^{T-4}=\{1, \ldots, 1, 1, 0, 0\}$. Let $W$ be the transmission channel and let $z$ be the Bhattacharyya parameter of the channel defined as
\begin{equation*}
(((W\overbrace{^{1})^{1})^{\ldots})^{1}}^{t-3 \footnotesize{\mbox{ times}}}.
\end{equation*}
Then, by using \eqref{eq:plusB}, we have that
\begin{equation*}
\begin{split}
Z_{T-5} & \le (2z-z^2)^4,\\
Z_{T-4} & \ge z^2 \sqrt{2-z^4} \sqrt{2-z^4(2-z^4)}. 
\end{split}
\end{equation*}
It is easy to check that, for any $z\in [0, 1]$,
\begin{equation}\label{eq:algcond}
(2z-z^2)^4\le z^2 \sqrt{2-z^4} \sqrt{2-z^4(2-z^4)},
\end{equation}
which implies that
\begin{equation*}
Z_{T-5}\le Z_{T-4}.
\end{equation*}
Consequently, as $v_{t_{T-5}} \geq K$, $v_{t_{T-4}} \geq K$. As a result, since $v_{t_{T-1}} < K$, $v_{t_{T-2}} < K$, and $v_{t_{T-3}} < K$, the node consists of frozen bits except for the last three bits which are information bits, i.e., it is a Type-II node.

\item Note that $\mathbf{b}^{2}=\{0, \ldots, 0, 1, 0\}$. By using the addition property \eqref{eq:add} and the left-swap property \eqref{eq:leftswap}, we obtain that $W_{2}\prec W_{i}$ for any $i\in\{3, 4, \ldots, T-1\}$. Hence, as $v_{t_2} < K$, $v_{t_i} < K$ for any $i\in\{3, 4, \ldots, T-1\}$. As $v_{t_0} \geq K$ and $v_{t_1} \geq K$, the node consists of information bits except for the first two bits which are frozen bits, i.e., it is a Type-III node.

\item Note that $\mathbf{b}^{4}=\{0, \ldots, 0, 1, 0, 0\}$. By using the addition property \eqref{eq:add} and the left-swap property \eqref{eq:leftswap}, we obtain that $W_{4}\prec W_{i}$ for any $i\in\{5, 6, \ldots, T-1\}$. Hence, as $v_{t_4} < K$, $v_{t_i} < K$ for any $i\in\{5, 6, \ldots, T-1\}$. Furthermore, note that $\mathbf{b}^{3}=\{0, \ldots, 0, 0, 1, 1\}$. Let $W$ be the transmission channel and let $z$ be the Bhattacharyya parameter of the channel defined as
\begin{equation*}
(((W\overbrace{^{0})^{0})^{\ldots})^{0}}^{t-3 \footnotesize{\mbox{ times}}}.
\end{equation*}
Then, by using \eqref{eq:plusB}, we have that
\begin{equation*}
\begin{split}
Z_{3} & \le (2z-z^2)^4,\\
Z_{4} & \ge z^2 \sqrt{2-z^4} \sqrt{2-z^4(2-z^4)}. 
\end{split}
\end{equation*}
Since \eqref{eq:algcond} holds for any $z\in [0, 1]$, we obtain that
\begin{equation*}
Z_{3}\le Z_{4}.
\end{equation*}
Consequently, as $v_{t_4} < K$, $v_{t_3} < K$. As a result, since $v_{t_0} \geq K$, $v_{t_1} \geq K$, and $v_{t_2} \geq K$, the node consists of information bits except for the first three bits which are frozen bits, i.e., it is a Type-IV node.

\item Note that $\mathbf{b}^{T-9}=\{1, \ldots, 1, 0, 1, 1\}$. By using the addition property \eqref{eq:add} and the left-swap property \eqref{eq:leftswap}, we obtain that $W_{i}\prec W_{T-9}$ for any $i\in\{0, 1, \ldots, T-10\}$. Hence, as $v_{t_{T-9}} \geq K$, $v_{t_i} \geq K$ for any $i\in\{0, 1, \ldots, T-10\}$. By using again the left-swap property \eqref{eq:leftswap}, we obtain that $W_{T-6}\prec W_{T-4}$ and $W_{T-7}\prec W_{T-4}$. By using again the addition property \eqref{eq:add}, we obtain that $W_{T-8}\prec W_{T-4}$. Hence, as $v_{t_{T-4}} \geq K$, $v_{t_i} \geq K$ for any $i\in\{T-6, T-7, T-8\}$. As a result, since $v_{t_{T-1}} < K$, $v_{t_{T-2}} < K$, $v_{t_{T-3}} < K$, and $v_{t_{T-5}} < K$, the node is a Type-V node.
\end{enumerate}
\end{proof}

The proofs for the identification of Rate-0, Rep, SPC, Rate-1, Type-I, Type-III, and Type-V nodes are based on stochastic degradation arguments. Consequently, these proofs are general and do not depend on the fact that the frozen bits are determined according to the value of the Bhattacharyya parameter. On the contrary, the proofs for Type-II and Type-IV nodes use the inequalities \eqref{eq:plusB} which are valid for Bhattacharyya parameters. However, let us point out that the strategy of the proof (use extremes of information combining bounds such as \eqref{eq:plusB} in order to compare the reliability of specific channels) is general. In order to prove a similar statement for different reliability measures, one would need to find bounds of the form \eqref{eq:plusB} for the desired reliability measure (e.g., mutual information, error probability). Let us further clarify that the proofs for Type-II and Type-IV nodes provide an ordering between the Bhattacharyya parameter of bit-channels. As such, they do not depend on the particular technique used to compute those Bhattacharyya parameters (Gaussian approximation \cite{trifonov_GA}, beta-expansion \cite{he_beta}, Monte Carlo simulation \cite{arikan}, etc.).
Let us also note that the Bhattacharyya parameter represents the typical performance metric employed for code construction \cite{tal,guo_bhat,vangala}.

It is also worth mentioning that since every node in the SC-based decoding tree represents a polar code constructed for a different channel \cite{arikan}, the results in this section are valid for all the nodes in any polar code of any length. Fig.~\ref{fig:logicallnew} shows the circuit required to generate the list of operations on-line for any node of length $T$, if Type-I, Type-II, Type-III, Type-IV, and Type-V nodes are considered in addition to Rate-0, Rep, SPC, and Rate-1 nodes. It can be seen that the circuit consists of ten comparators, nine NOT gates, and fourteen AND gates, in order to identify all the special nodes.

\begin{figure}
  \centering
  \begin{tikzpicture}[circuit logic US, tiny circuit symbols, scale=.75, thick]

\foreach \x in {0,1.5,...,6}
{
\draw (\x,0) -- ++(.75,-.75) -- ++(0,-.5) -- ++(-.75,-.75) -- ++(0,.75) -- ++(.25,.25) -- ++(-.25,.25) -- ++(0,.75);

\node (A) at (\x+.15,-.5) {\tiny $A$};
\node (B) at (\x+.15,-1.5) {\tiny $B$};
\node (C) at (\x+.6,-1) {\tiny $C$};

\draw (\x,-.5) -- ++(-.25,0) -- ++(0,.75);
\draw (\x,-1.5) -- ++(-.25,0) -- ++(0,-.75);
\draw (\x+.75,-1) -- ++(.25,0) -- ++(0,-1.5);
}

\draw (-.75,-2.25) -- (5.75,-2.25);

\node at (-1,-2.25) {$K$};

\node at (-.25,.5) {$v_{t_0}$};
\node at (1.25,.5) {$v_{t_1}$};
\node at (2.75,.5) {$v_{t_2}$};
\node at (4.25,.5) {$v_{t_4}$};
\node at (5.75,.5) {$v_{t_{T-9}}$};

\node at (1,-2.75) {$s_{t_0}$};
\node at (2.5,-2.75) {$s_{t_1}$};
\node at (4,-2.75) {$s_{t_2}$};
\node at (5.5,-2.75) {$s_{t_4}$};
\node at (7,-2.75) {$s_{t_{T-9}}$};

\end{tikzpicture}
  \begin{tikzpicture}[circuit logic US, tiny circuit symbols, scale=.75, thick]

\foreach \x in {0,1.5,...,6}
{
\draw (\x,0) -- ++(.75,-.75) -- ++(0,-.5) -- ++(-.75,-.75) -- ++(0,.75) -- ++(.25,.25) -- ++(-.25,.25) -- ++(0,.75);

\node (A) at (\x+.15,-.5) {\tiny $A$};
\node (B) at (\x+.15,-1.5) {\tiny $B$};
\node (C) at (\x+.6,-1) {\tiny $C$};

\draw (\x,-.5) -- ++(-.25,0) -- ++(0,.75);
\draw (\x,-1.5) -- ++(-.25,0) -- ++(0,-.75);
\draw (\x+.75,-1) -- ++(.25,0) -- ++(0,-1.5);
}

\draw (-.75,-2.25) -- (5.75,-2.25);

\node at (-1,-2.25) {$K$};

\node at (-.25,.5) {$v_{t_{T-5}}$};
\node at (1.25,.5) {$v_{t_{T-4}}$};
\node at (2.75,.5) {$v_{t_{T-3}}$};
\node at (4.25,.5) {$v_{t_{T-2}}$};
\node at (5.75,.5) {$v_{t_{T-1}}$};

\node at (1,-2.75) {$s_{t_{T-5}}$};
\node at (2.5,-2.75) {$s_{t_{T-4}}$};
\node at (4,-2.75) {$s_{t_{T-3}}$};
\node at (5.5,-2.75) {$s_{t_{T-2}}$};
\node at (7,-2.75) {$s_{t_{T-1}}$};

\end{tikzpicture}
  \begin{tikzpicture}[circuit logic US, tiny circuit symbols, scale=1, thick]
\node [and gate, inputs=nn] (and1) at (0,0) {};
\node [and gate, inputs=nn] (and2) at (0,-.75) {};
\node [and gate, inputs=nn] (and3) at (1,-1.25) {};
\node [and gate, inputs=nn] (and4) at (1,-2) {};
\node [and gate, inputs=nn] (and5) at (0,-2.075) {};
\node [not gate] (not1) at (-1.25,.375) {};
\node [not gate] (not2) at (-1.25,-.825) {};
\node [not gate] (not3) at (-1.25,-1.75) {};

\node (s0) at ([xshift=-2cm,yshift=.875cm]and1.input 2) {$s_{t_0}$};
\node (s1) at ([xshift=-2cm]and1.input 2) {$s_{t_1}$};
\node (s2) at ([xshift=-2cm,yshift=-1.25cm]and1.input 2) {$s_{t_2}$};
\node (s4) at ([xshift=-2cm,yshift=-2.075cm]and1.input 2) {$s_{t_4}$};
\draw (s0) --++(0:.6cm) |- (not1.input);
\draw (not1.output) --++(0:.25cm) |- (and1.input 1);
\draw (s1) -- (and1.input 2);
\draw (s1) --++(0:.6cm) |- (not2.input);
\draw (not2.output) -- (and2.input 2);
\draw (not1.output) --++(0:.25cm) |- (and2.input 1);
\draw (s2) -- (and3.input 2);
\draw (and2.output) --++(0:.2cm) |- (and3.input 1);
\draw (s2) --++(0:.6cm) |- (not3.input);
\draw (and2.output) --++(0:.2cm) |- (and4.input 1);
\draw (not3.output) --++(0:.2cm) |- (and5.input 1);
\draw (and5.output) -- (and4.input 2);
\draw (s4) -- (and5.input 2);

\node (out-rate1) at ([xshift=5.5cm]s0) {Rate-1};
\node (out-spc) at ([xshift=2.75cm]and1.output) {SPC};
\node (out-typeiii) at ([xshift=2cm]and3.output) {Type-III};
\node (out-typeiv) at ([xshift=2cm]and4.output) {Type-IV};

\draw (s0) -- (out-rate1);
\draw (and1.output) -- (out-spc);
\draw (and3.output) -- (out-typeiii);
\draw (and4.output) -- (out-typeiv);

%-----

\node [not gate] (not6) at (-1.25,-3) {};
\node [not gate] (not7) at (-1.25,-4) {};
\node [not gate] (not8) at (-1.25,-5) {};
\node [not gate] (not9) at (-1.25,-6) {};
\node [not gate] (not10) at (-1.25,-7) {};
\node [not gate] (not11) at (-1.25,-8) {};

\node [and gate, inputs=nn] (and6) at (0,-3.5) {};
\node [and gate, inputs=nn] (and7) at (0,-4.5) {};
\node [and gate, inputs=nn] (and8) at (0,-5.5) {};
\node [and gate, inputs=nn] (and9) at (0,-6.5) {};
\node [and gate, inputs=nn] (and10) at (0,-7.5) {};
\node [and gate, inputs=nn] (and11) at (1,-6.425) {};
\node [and gate, inputs=nn] (and12) at (1,-4.575) {};
\node [and gate, inputs=nn] (and13) at (1,-5.575) {};
\node [and gate, inputs=nn] (and14) at (2,-3.575) {};

\node (s6) at ([xshift=-1cm]not6.input) {$s_{t_{T-9}}$};
\node (s7) at ([xshift=-2.2cm]and6.input 2) {$s_{t_{T-5}}$};
\node (s8) at ([xshift=-1cm]not8.input) {$s_{t_{T-4}}$};
\node (s9) at ([xshift=-2.2cm]and8.input 2) {$s_{t_{T-3}}$};
\node (s10) at ([xshift=-2.2cm]and9.input 1) {$s_{t_{T-2}}$};
\node (s11) at ([xshift=-2.2cm]and10.input 2) {$s_{t_{T-1}}$};

\draw (s6) -- (not6.input);
\draw (not6.output) --++(0:.25cm) |- (and6.input 1);
\draw (s7) -- (and6.input 2);
\draw (s7) --++(0:.8cm) |- (not7.input);

\draw (s8) -- (not8.input);
\draw (not7.output) --++(0:.25cm) |- (and7.input 1);

\draw (s9) -- (and8.input 2);
\draw (s9) --++(0:.8cm) |- (and7.input 2);
\draw (s9) --++(0:.8cm) |- (not9.input);
\draw (not8.output) --++(0:.25cm) |- (and8.input 1);
\draw (s10) --++(0:.9cm) |- (not10.input);

\fill (s9) --++(0:.8cm) circle (.05cm);

\draw (not9.output) --++(0:1.25cm) |- (and11.input 1);
\draw (not10.output) --++(0:.25cm) |- (and10.input 1);

\draw (s10) -- (and9.input 1);
\draw (s11) --++(0:.8cm) |- (and9.input 2);
\draw (s11) -- (and10.input 2);
\draw (s11) --++(0:.8cm) |- (not11.input);

\draw (and9.output) -- (and11.input 2);
\draw (and7.output) -- (and12.input 1);
\draw (and8.output) -- (and13.input 1);
\draw (and9.output) --++(0:.3cm) |- (and12.input 2);
\draw (and9.output) --++(0:.3cm) |- (and13.input 2);
\draw (and6.output) -- (and14.input 1);
\draw (and13.output) --++(0:.25cm) |- (and14.input 2);

\node (out-typev) at ([xshift=1cm]and14.output) {Type-V};
\node (out-typeii) at ([xshift=2cm]and12.output) {Type-II};
\node (out-typei) at ([xshift=2cm]and11.output) {Type-I};
\node (out-rep) at ([xshift=2.85cm]and10.output) {Rep};
\node (out-rate0) at ([xshift=4.15cm]not11.output) {Rate-0};

\draw (and14.output) -- (out-typev);
\draw (and12.output) -- (out-typeii);
\draw (and11.output) -- (out-typei);
\draw (and10.output) -- (out-rep);
\draw (not11.output) -- (out-rate0);

\end{tikzpicture}
  \caption{Efficient generation of the list of operations on hardware considering new nodes.}
  \label{fig:logicallnew}
\end{figure}
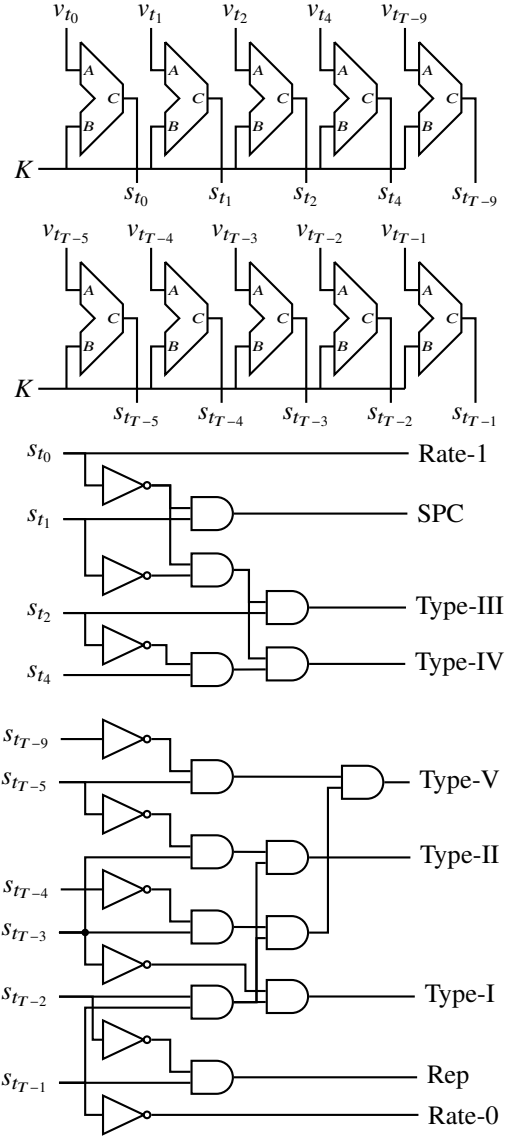

\section{Decoder Architecture} \label{sec:decarch}

As a proof of concept, a decoder architecture implementing the proposed technique has been designed. It implements the layered partitioned SCL (LPSCL) decoding algorithm detailed in \cite{hashemi_TCOM} and the Fast-SSCL-SPC algorithm introduced in \cite{hashemi_FSSCL_TSP}, along with the memory-reduction techniques proposed in \cite{hashemi_JETCAS}. The LPSCL decoder decreases the memory requirements of standard SCL decoding by dividing the SC decoding tree in different partitions; the bottom part of the SC decoding tree belonging to each partition is decoded with SCL with a list size $L_{\max}$. When information needs to be passed between partitions, i.e. at the top stages of the tree, only $L_t<L_{\max}$ candidate codewords are passed, with $L_t$ decreasing progressively as the stage $t$ increases. The Fast-SSCL-SPC algorithm is applied to the lower stages of the tree, where $L_{\max}$ candidates are considered.

\begin{figure}[t!]
  \centering
  \begin{tikzpicture}[scale=0.9, thick]
\scriptsize

\draw [fill=lightgray] (-.5,.5) rectangle ++(4,-1.75);
\draw [fill=lightgray] (4.5,.5) rectangle ++(4,-1.75);

\draw [fill=white] (0,0) rectangle ++(1,-.75) node[pos=.5] {$\text{PE}_0$};
\node at (1.5,-.375) {$\cdots$};
\draw [fill=white] (1.9,0) rectangle ++(1.2,-.75) node[pos=.5] {$\text{PE}_{N_{\text{PE}}-1}$};
\node at (-.25,.25) {$0$};

\node at (4,-.375) {$\cdots$};

\draw [fill=white] (5,0) rectangle ++(1,-.75) node[pos=.5] {$\text{PE}_0$};
\node at (6.5,-.375) {$\cdots$};
\draw [fill=white] (6.9,0) rectangle ++(1.2,-.75) node[pos=.5] {$\text{PE}_{N_{\text{PE}}-1}$};
\node at (5,.25) {$L_{\max}\!\!-\!\!1$};

\draw (0,2) rectangle ++(8,-1) node[pos=.5,align=center,text width=8cm] {Path Memory \\ LLR Memory \\ PM Memory};

\draw (1.5,1) -- ++(0,-.75);
\draw (1.45,.7) -- ++(.1,.1);
\draw [->] (1.5,.25) -- ++(-1,0) -- ++(0,-.25);
\draw [->] (1.5,.25) -- ++(1,0) -- ++(0,-.25);

\draw (6.5,1) -- ++(0,-.75);
\draw (6.45,.7) -- ++(.1,.1);
\draw [->] (6.5,.25) -- ++(-1,0) -- ++(0,-.25);
\draw [->] (6.5,.25) -- ++(1,0) -- ++(0,-.25);

\draw (.5,-.75) -- ++(0,-.25) -- ++(1,0);
\draw [->] (2.5,-.75) -- ++(0,-.25) -- ++(-1,0) -- ++(0,-1);
\draw (0,-2) rectangle ++(3,-.5) node[pos=.5] {PM Calculation};
\draw (1.45,-1.55) -- ++(.1,.1);
\draw [->] (1.5,-2.5) -- ++(0,-.25);

\draw (1.5,-1.65) -- ++(-2.25,0);

\node at (4,-2.25) {$\cdots$};

\draw (5.5,-.75) -- ++(0,-.25) -- ++(1,0);
\draw [->] (7.5,-.75) -- ++(0,-.25) -- ++(-1,0) -- ++(0,-1);
\draw (5,-2) rectangle ++(3,-.5) node[pos=.5] {PM Calculation};
\draw (6.45,-1.55) -- ++(.1,.1);
\draw [->] (6.5,-2.5) -- ++(0,-.25);

\draw (1.5,-1.75) -- ++(2.5,0);
\draw [->] (6.5,-1.75) -- ++(-2.5,0) -- ++(0,2.75);

\draw (.5,-2.75) rectangle ++(7,-.5) node[pos=.5] {PM Sorting};
\draw [->] (4,-3.25) -- ++(0,-.25) -- ++(-4.75,0) -- ++(0,5) -- ++(.75,0);

\end{tikzpicture}
  \caption{Decoder architecture.}
  \label{fig:arch}
\end{figure}
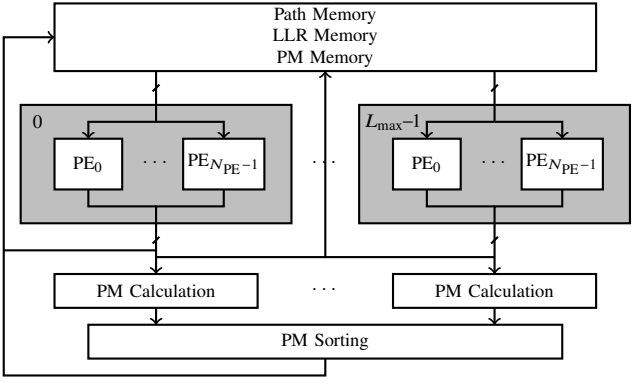

Fig.~\ref{fig:arch} shows the architecture of the proposed decoder. It is based on a semi-parallel SCL decoder architecture, where $L_{\max}$ sets of $N_{\text{PE}}$ processing elements (PEs) are instantiated in parallel, implementing (\ref{eq:Ffunc2HW}) and (\ref{eq:Gfunc2}). Each set works on a different candidate codeword, as explained in Section~\ref{sec:prel:SCDec}. Different candidate codewords are created whenever one or more information bits are estimated.
Each set of PEs relies on a dedicated memory to store the internal LLR values relative to all stages of the SC decoding tree. LLR values are quantized with $Q_{\text{LLR}}$ bits, and represented with sign and magnitude. Each stage of the SC decoding tree requires the storage of $2^{t-1}$ LLR values. However, given the limited number of PEs instantiated, the LLR memory is split in high stage and low stage memories. The high stage memory stores LLR values of stages with nodes of size greater than $N_{\text{PE}}$: at stage $t$, where $2^t>2N_{\text{PE}}$, a total of $2^t/(2N_{\text{PE}})$ decoding steps are needed to descend to the lower tree level. The depth of the high stage memory is $\sum_{j=\log_2 N_{\text{PE}}+1}^{n-1} 2^j/N_{\text{PE}} = N/N_{\text{PE}}-2$, while it is $Q_{\text{LLR}}\times N_{\text{PE}}$ wide. The low stage memory stores LLR values for stages where $2^t\le 2N_{\text{PE}}$, and it is $Q_{\text{LLR}}$ bits wide, while its depth is $\sum_{j=0}^{\log_2 N_{\text{PE}} - 1} N_{\text{PE}}/2^j = 2N_{\text{PE}}-2$. High and low stage memory words are rewritten when a node belonging to the same stage $t$ is traversed. $L_{\max}$ different instantiations of both high and low stage memories are required. $L_{\max}$ separate path memories store the hard bit estimates (\ref{eq:beta}) for all the tree stages as well, updating them every time that a bit is estimated. PMs, that identify the likelihood of a candidate codeword (or path) to be correct, are incremented every time a bit is estimated differently from the sign of the LLR value associated to it. They are sorted in PM memory before and after the estimation of an information bit, in order to identify the $L_{\max}$ surviving paths out of the $2L_{\max}$ created. When none of the paths coming from the splitting of a particular candidate codeword survives, all stages of its LLR memory are overwritten, along with the bit estimate and PM memories.

This baseline architecture has been modified to implement the LPSCL decoder. The bottom stages of the SC decoding tree are left unchanged, and decoded with a list size $L_{\max}$. Given the partitioning factor $P$, the top $\log_2P$ stages rely on a smaller list size $L_t$, with $n-\log_2P< t \le n$, and $L_t\ge L_{t+1}$. Consequently, only $L_t$ LLR memories are instantiated in the upper stages, reducing the LLR memory requirements for each upper stage of a factor $\frac{L_{\max}-L_t}{L_{\max}}$, as shown in Fig.~\ref{fig:llrmem}. Depending on the number of instantiated PEs and on the partitioning factor, the high and/or low stage memories might need to be separated into different memory structures, each part belonging to a different layer of LPSCL and thus instantiated a different number of times, depending on $L_t$. Since the number of surviving paths is reduced from $L_{\max}$ to $L_t$ when ascending the decoding tree above stage $n-\log_2P$, the $L_{\max}-L_t$ candidate codewords with the highest PMs need to be discarded. In the baseline architecture, PMs are sorted only when an information bit is estimated, i.e. when the paths split. However, in the proposed architecture the PMs need to be sorted also when $i \mod (N/P) = 0$, where $i$ is the index of the codeword bit that needs to be estimated, and $\mod$ represents the modulo operation. The decoding of a bit with such an index $i$ identifies the completion of the decoding of a subtree of size $N/P$, and the need to transfer information to the upper tree stages, where $L_t<L_{\max}$. The sorting of PMs allows the most reliable paths, their LLR values, and their hard bit estimates to be transferred between partitions through the memory copy mechanism addressed in Fig.~\ref{fig:llrmem}.

\begin{figure*}[t!]
  \centering
  \begin{tikzpicture}[scale=1, thick]

\draw (0,1) rectangle ++(6,-.5) node[pos=.5] {Copy};
\draw (0,.25) rectangle ++(2.5,-1) node[pos=.5,align=center,text width=2.5cm] {LLR Memory \\ $0$};
\node at (3,-.25) {$\cdots$};
\draw (3.5,.25) rectangle ++(2.5,-1) node[pos=.5,align=center,text width=2.5cm] {LLR Memory \\ $L_{t+1}\!\!-\!\!1$};

\draw [->] (1,.5) -- ++(0,-.25);
\draw [->] (1.5,.25) -- ++(0,.25);
\draw [->] (4.5,.5) -- ++(0,-.25);
\draw [->] (5,.25) -- ++(0,.25);

\node at (8.25,-.25) {Stage $t+1$};

\draw [<->] (0,-.5) -- ++(-.25,0) -- ++(0,-1.75) -- ++(.25,0);
\draw [<->] (3.5,-.5) -- ++(-.25,0) -- ++(0,-1.75) -- ++(.25,0);

\draw (0,-1.25) rectangle ++(9.5,-.5) node[pos=.5] {Copy};
\draw (0,-2) rectangle ++(2.5,-1) node[pos=.5,align=center,text width=2.5cm] {LLR Memory \\ $0$};
\node at (3,-2.5) {$\cdots$};
\draw (3.5,-2) rectangle ++(2.5,-1) node[pos=.5,align=center,text width=2.5cm] {LLR Memory \\ $L_{t+1}\!\!-\!\!1$};
\node at (6.5,-2.5) {$\cdots$};
\draw (7,-2) rectangle ++(2.5,-1) node[pos=.5,align=center,text width=2.5cm] {LLR Memory \\ $L_{t}\!\!-\!\!1$};

\draw [->] (1,-1.75) -- ++(0,-.25);
\draw [->] (1.5,-2) -- ++(0,.25);
\draw [->] (4.5,-1.75) -- ++(0,-.25);
\draw [->] (5,-2) -- ++(0,.25);
\draw [->] (8,-1.75) -- ++(0,-.25);
\draw [->] (8.5,-2) -- ++(0,.25);

\node at (11.75,-2.5) {Stage $t$};

\draw [<->] (0,-2.75) -- ++(-.25,0) -- ++(0,-1.75) -- ++(.25,0);
\draw [<->] (3.5,-2.75) -- ++(-.25,0) -- ++(0,-1.75) -- ++(.25,0);
\draw [<->] (7,-2.75) -- ++(-.25,0) -- ++(0,-1.75) -- ++(.25,0);

\draw (0,-3.5) rectangle ++(13,-.5) node[pos=.45] {Copy};
\draw (0,-4.25) rectangle ++(2.5,-1) node[pos=.5,align=center,text width=2.5cm] {LLR Memory \\ $0$};
\node at (3,-4.75) {$\cdots$};
\draw (3.5,-4.25) rectangle ++(2.5,-1) node[pos=.5,align=center,text width=2.5cm] {LLR Memory \\ $L_{t+1}\!\!-\!\!1$};
\node at (6.5,-4.75) {$\cdots$};
\draw (7,-4.25) rectangle ++(2.5,-1) node[pos=.5,align=center,text width=2.5cm] {LLR Memory \\ $L_{t}\!\!-\!\!1$};
\node at (10,-4.75) {$\cdots$};
\draw (10.5,-4.25) rectangle ++(2.5,-1) node[pos=.5,align=center,text width=2.5cm] {LLR Memory \\ $L_{\max}\!\!-\!\!1$};

\draw [->] (1,-4) -- ++(0,-.25);
\draw [->] (1.5,-4.25) -- ++(0,.25);
\draw [->] (4.5,-4) -- ++(0,-.25);
\draw [->] (5,-4.25) -- ++(0,.25);
\draw [->] (8,-4) -- ++(0,-.25);
\draw [->] (8.5,-4.25) -- ++(0,.25);
\draw [->] (11.5,-4) -- ++(0,-.25);
\draw [->] (12,-4.25) -- ++(0,.25);

\node [rotate=90] at (-.75,-2.5) {Upper/Lower Stage Transfer};

\draw (-.5,-2.5) -- ++(0,1) -- ++(.15,0);
\draw (-.5,-2.5) -- ++(0,-1) -- ++(.15,0);

\end{tikzpicture}
  \caption{LPSCL LLR memory structure.}
  \label{fig:llrmem}
\end{figure*}
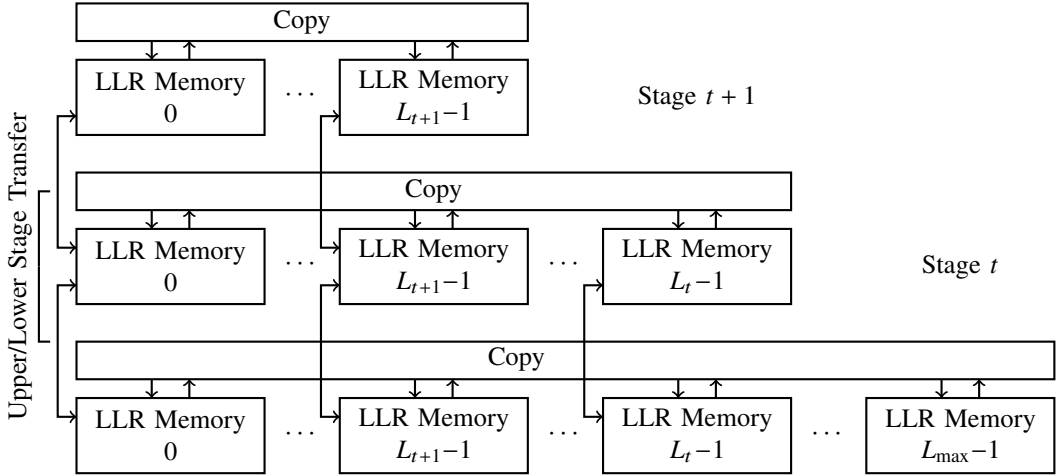

The implementation of the Fast-SSCL-SPC algorithm requires more substantial modifications, that have been detailed in \cite{hashemi_FSSCL_TSP}. The hard bit estimate memory and path memories are updated according to different values depending on the node type, along with PMs. This requires different parallel instantiations of the PM computation logic, as shown in Fig.~\ref{fig:pm}. More complex routing and selection logic are necessary to update memories, since multiple concurrent values need to be updated and propagated through the hard bit estimates memory structure. A sorter module for LLR values is needed in Rate-1 and SPC nodes, to identify the order with which bits are estimated: the disruption of the sequential bit estimation order that SC is based on leads to additional complexity in memory updates and control logic.

\begin{figure*}
  \centering
  \begin{tikzpicture}[scale=1.25, thick]

\draw (0,0) rectangle ++(1,-.75) node[pos=.5] {Rate-1};
\draw (1.5,0) rectangle ++(1,-.75) node[pos=.5] {Rep};
\draw (3,0) rectangle ++(1,-.75) node[pos=.5] {SPC};
\draw (5,0) rectangle ++(1,-.75) node[pos=.5] {Rate-0};
\draw (7,0) rectangle ++(1,-.75) node[pos=.5] {Rate-1};
\draw (8.5,0) rectangle ++(1,-.75) node[pos=.5] {Rep};
\draw (10,0) rectangle ++(1,-.75) node[pos=.5] {SPC};

\draw [dashed] (-.5,.5) rectangle ++(5,-2.75);
\draw [dashed] (6.5,.5) rectangle ++(5,-2.75);

\draw (0,-1.25) -- ++(4,0) -- ++(-.5,-.5) -- ++(-3,0) -- ++(-.5,.5);
\draw (7,-1.25) -- ++(4,0) -- ++(-.5,-.5) -- ++(-3,0) -- ++(-.5,.5);

\draw [->] (5.5,.75) -- ++(0,-.75);
\draw (5.45,.45) -- ++(.1,.1);
\draw [->] (5.5,.25) -- ++(-2,0) -- ++(0,-.25);
\draw [->] (5.5,.25) -- ++(2,0) -- ++(0,-.25);
\draw [->] (3.5,.25) -- ++(-1.5,0) -- ++(0,-.25);
\draw [->] (7.5,.25) -- ++(1.5,0) -- ++(0,-.25);
\draw [->] (2,.25) -- ++(-1.5,0) -- ++(0,-.25);
\draw [->] (9,.25) -- ++(1.5,0) -- ++(0,-.25);

\draw [->] (.5,-.75) -- ++(0,-.5);
\draw [->] (2,-.75) -- ++(0,-.5);
\draw [->] (3.5,-.75) -- ++(0,-.5);
\draw [->] (7.5,-.75) -- ++(0,-.5);
\draw [->] (9,-.75) -- ++(0,-.5);
\draw [->] (10.5,-.75) -- ++(0,-.5);

\draw [->] (5.5,-.75) -- ++(0,-.25) -- ++(-1.75,0) -- ++(0,-.25);
\draw [->] (5.5,-1) -- ++(1.75,0) -- ++(0,-.25);

\draw [->] (2,-1.75) -- ++(0,-.75);
\draw [->] (9,-1.75) -- ++(0,-.75);

\node at (5.5,1) {LLR Values};
\node at (2,-2.75) {PM};
\node at (9,-2.75) {PM};
\node at (0,.25) {$\text{bit} = 0$};
\node at (11,.25) {$\text{bit} = 1$};

\end{tikzpicture}
  \caption{PM calculation for Fast-SSCL-SPC.}
  \label{fig:pm}
\end{figure*}
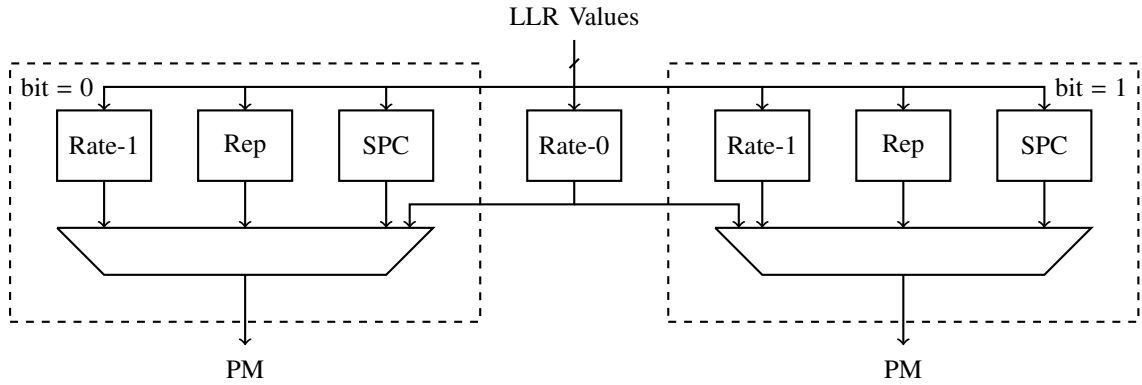

Aside from the logic needed to perform the calculations for special node PM update and bit estimations, the decoder needs to know at which point in the SC tree the special nodes are found, and what is their type. This information is used to identify the number of clock cycles needed for the decoding of a particular node, and which of the different parallel PM, path, and LLR updates is memorized. In \cite{hashemi_FSSCL_TSP}, the proposed decoder architecture relied on an off-line compiler to obtain the sequence of special nodes, their size, and the stage at which they are encountered. These informations differ for every code supported by the decoder, and need to be stored in a memory. Note that the frozen and information bit sequence can be either stored in a memory, as supposed by most decoder architectures in literature, or computed on-line given the bit-channel relative reliability vector and the desired code rate, as proposed in \cite{Condo_SIPS}. This approach is significantly more efficient in case of multi-code decoders, and is facilitated by nested reliability vectors as those selected for the 5G eMBB control channel \cite{3gpp_polarSequence}. This is the approach taken in both the baseline and the modified architectures in this paper, by comparing each entry of the relative reliability vector $\mathbf{v}$ to the desired $K$ in order to obtain $\mathbf{s}$.

The control unit of the modified architecture implements the proposed special node on-line identification, based on the relative reliability vector $\mathbf{v}$ and $K$. Fig.~\ref{fig:logicall} shows the simple logic needed to identify the considered special nodes. Given the low complexity of the node identification circuit, the structure is instantiated at every decoding tree stage $t$, separately at every partition identified by LPSCL, to reduce the amount of multiplexing needed at the inputs and the possible increase in the system critical path.
The logic pictured in Fig.~\ref{fig:logicall} is inserted within a finite state machine (FSM) in the decoder control unit to identify the correct decoding phase, through two main control signals, \texttt{NodeType} and \texttt{NodeSize}. A maximum \texttt{NodeSize} value for each \texttt{NodeType} is selected at design time, to limit the additional complexity and critical path degradation. 
\begin{itemize}
 \item While the general node type can be identified easily through the proposed identification, different decoding phases are foreseen within each special node. Thus, \texttt{NodeType} foresees subtypes in the special node. While the Rate-0 node is a standalone node type, the Rate-1 node is divided into three subtypes: one phase is assigned to the fetching and sorting of the LLR values, a second to the estimation of the bits associated to the least reliable LLR values, and the third for the hard-decision on the remainder of the bits. The Rep node is divided in two subtypes, one for the frozen bits and one for the information bit. Finally, SPC nodes foresee four subtypes: one for the concurrent fetching and sorting of LLR values and frozen bit selection, one for the bit estimations, one for the hard decision on the remaining bits, and one for the parity correction. The \texttt{NodeType} signal is thus influenced not only by the result of the logic in Fig.~\ref{fig:logicall}, but also by the number of estimated bits within the special node, the stage $t$, and the current \texttt{NodeType} subtype.
 \item The control unit identifies the size of the special node \texttt{NodeSize} as $2^t$, given the current SC decoding tree stage $t$. This information is used to update the index $i$ of the codeword bit to be estimated. The index $i$ is usually updated once a leaf node has been reached and the corresponding bit estimated, but during the decoding of special nodes, it is kept fixed pointing at the first bit of the node. Once the decoding is terminated, the index is updated as $i+\texttt{NodeSize}$.
\end{itemize}

\section{Hardware Implementation Results} \label{sec:results}

The proposed decoder architecture has been described in VHDL and synthesized in TSMC 65~nm CMOS technology, at the operating conditions defined by the NCCOM corner, i.e. $1.2$~V core voltage and a temperature of $298$~K. Two versions of the decoder have been implemented: one considering the proposed special node identification technique, and one based on the off-line identification and storage used in \cite{hashemi_FSSCL_TSP}. Both decoders target the 5G polar code with a code length $N=1024$ \cite{3gpp_polarSequence}, rely on a partitioning factor $P=4$, and make use of $64$ parallel PEs. The bottom part of the SC decoding tree is decoded with a list size $L_{\max}=4$, while for the upper stages $L_{10}=L_{9}=2$. Fig.~\ref{fig:performance} shows the frame error rate (FER) and bit error rate (BER) performance of the LPSCL decoder used in this paper in comparison with SCL decoding with $L=4$. The curves in Fig.~\ref{fig:performance} are provided for the code rates of $\{\frac{1}{12},\frac{1}{6},\frac{1}{3},\frac{1}{2},\frac{2}{3}\}$. It can be seen that LPSCL decoding incurs negligible FER and BER performance loss with respect to SCL for all considered rates. It should be noted that the introduction of the proposed technique to infer the list of operations on the fly does not change the FER or BER performance of the decoder in comparison with the same memory-based decoder.

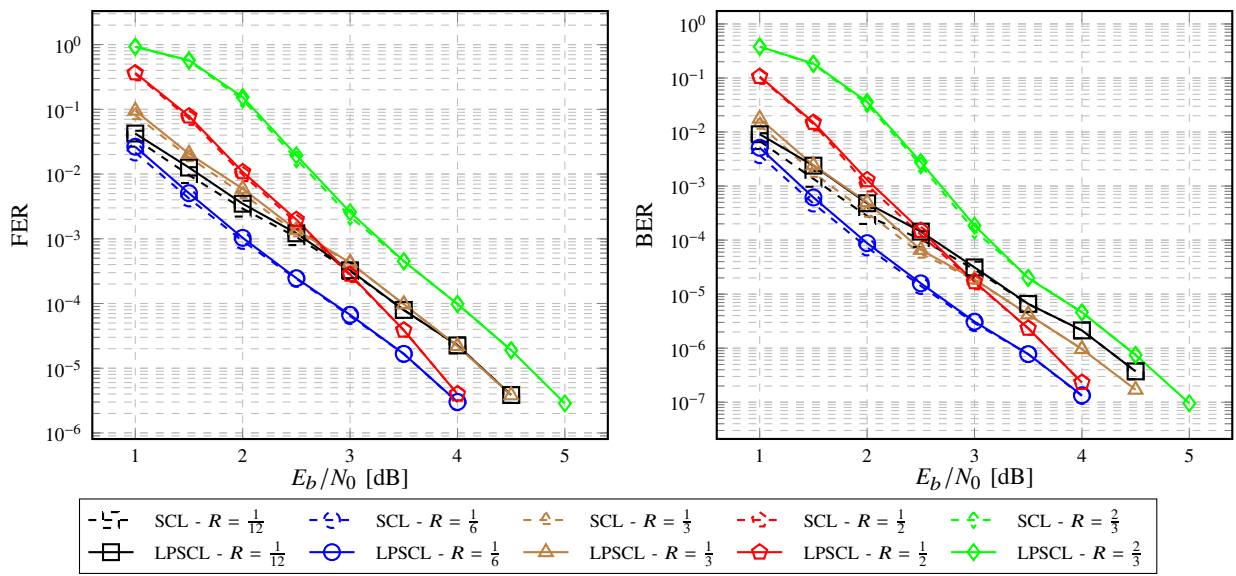
\begin{figure*}
  \centering
  \begin{tikzpicture}
  \pgfplotsset{
    label style = {font=\fontsize{9pt}{7.2}\selectfont},
    tick label style = {font=\fontsize{7pt}{7.2}\selectfont}
  }

\begin{axis}[
	scale = 1,
    ymode=log,
    xlabel={$E_b/N_0$ [\text{dB}]}, xlabel style={yshift=0.8em},
    ylabel={FER}, ylabel style={yshift=-0.75em},
    grid=both,
    ymajorgrids=true,
    xmajorgrids=true,
    grid style=dashed,
%    width=0.5\columnwidth, height=7cm,
    thick,
    mark size=3,
    legend style={
      anchor={center},
      cells={anchor=west},
      column sep= 2mm,
      font=\fontsize{7pt}{7.2}\selectfont,
    },
    legend to name=perf-legend1kL4,
    legend columns=5,
]

\addplot[
    color=black,
    mark=square,
    thick,
    mark size=3,
    dashed,
]
table {
1 0.0352
1.5 0.00962371
2 0.00292458
2.5 0.00106002
3 0.000317617
3.5 7.92321e-05
4 2.25283e-05
4.5 3.85473e-06
};
\addlegendentry{SCL - $R = \frac{1}{12}$}

\addplot[
    color=blue,
    mark=o,
    thick,
    mark size=3,
    dashed,
]
table {
1 0.0215
1.5 0.00420292
2 0.000926836
2.5 0.000239625
3 6.4039e-05
3.5 1.66664e-05
4 2.97072e-06
};
\addlegendentry{SCL - $R = \frac{1}{6}$}

\addplot[
    color=brown,
    mark=triangle,
    thick,
    mark size=3,
    dashed,
]
table {
1 0.0812
1.5 0.0182
2 0.00490822
2.5 0.00125389
3 0.000424899
3.5 9.74888e-05
4 2.2e-05
};
\addlegendentry{SCL - $R = \frac{1}{3}$}

\addplot[
    color=red,
    mark=pentagon,
    thick,
    mark size=3,
    dashed,
]
table {
1 0.352
1.5 0.0725
2 0.00978474
2.5 0.00177749
3 0.000270537
3.5 3.87083e-05
4 3.99636e-06
};
\addlegendentry{SCL - $R = \frac{1}{2}$}

\addplot[
    color=green,
    mark=diamond,
    thick,
    mark size=3,
    dashed,
]
table {
1 0.9258
1.5 0.5609
2 0.1418
2.5 0.017
3 0.00219072
3.5 0.000450479
4 9.82839e-05
4.5 1.89824e-05
};
\addlegendentry{SCL - $R = \frac{2}{3}$}

\addplot[
    color=black,
    mark=square,
    thick,
    mark size=3,
]
table {
1 0.0423
1.5 0.0126
2 0.00348117
2.5 0.00121459
3 0.000325961
3.5 7.96135e-05
4 2.25283e-05
4.5 3.85473e-06
};
\addlegendentry{LPSCL - $R = \frac{1}{12}$}

\addplot[
    color=blue,
    mark=o,
    thick,
    mark size=3,
]
table {
1 0.0264
1.5 0.00506047
2 0.00103361
2.5 0.000246176
3 6.68715e-05
3.5 1.66664e-05
4 3.01186e-06
};
\addlegendentry{LPSCL - $R = \frac{1}{6}$}

\addplot[
    color=brown,
    mark=triangle,
    thick,
    mark size=3,
]
table {
1 0.0953
1.5 0.0206
2 0.00592768
2.5 0.00137172
3 0.000426052
3.5 9.74888e-05
4 2.2e-05
4.5 3.81068e-06
};
\addlegendentry{LPSCL - $R = \frac{1}{3}$}

\addplot[
    color=red,
    mark=pentagon,
    thick,
    mark size=3,
]
table {
1 0.3656
1.5 0.0787
2 0.0108
2.5 0.00197394
3 0.000282311
3.5 3.87083e-05
4 3.99636e-06
};
\addlegendentry{LPSCL - $R = \frac{1}{2}$}

\addplot[
    color=green,
    mark=diamond,
    thick,
    mark size=3,
]
table {
1 0.9274
1.5 0.574
2 0.1542
2.5 0.0195
3 0.00257679
3.5 0.000450479
4 9.82839e-05
4.5 1.89824e-05
5 2.86469e-06
};
\addlegendentry{LPSCL - $R = \frac{2}{3}$}

\end{axis}
\end{tikzpicture}
  \begin{tikzpicture}
  \pgfplotsset{
    label style = {font=\fontsize{9pt}{7.2}\selectfont},
    tick label style = {font=\fontsize{7pt}{7.2}\selectfont},
  }

\begin{axis}[
	scale = 1,
    ymode=log,
    xlabel={$E_b/N_0$ [\text{dB}]}, xlabel style={yshift=0.8em},
    ylabel={BER}, ylabel style={yshift=-0.75em},%
    grid=both,
    ymajorgrids=true,
    xmajorgrids=true,
%    width=0.5\columnwidth, height=7.0cm,
    grid style=dashed,
    thick,
    mark size=3,
]

\addplot[
    color=black,
    mark=square,
    thick,
    mark size=3,
    dashed,
]
table {
1 0.00678588
1.5 0.00136091
2 0.000279039
2.5 0.000104256
3 2.85108e-05
3.5 6.54364e-06
4 2.13356e-06
4.5 3.72775e-07
};
%\addlegendentry{SCL - $R = \frac{1}{12}$}

\addplot[
    color=blue,
    mark=o,
    thick,
    mark size=3,
    dashed,
]
table {
1 0.00372941
1.5 0.000474682
2 7.31655e-05
2.5 1.41802e-05
3 2.92319e-06
3.5 7.77439e-07
4 1.31236e-07
};
%\addlegendentry{SCL - $R = \frac{1}{6}$}

\addplot[
    color=brown,
    mark=triangle,
    thick,
    mark size=3,
    dashed,
]
table {
1 0.0136378
1.5 0.00176012
2 0.000321265
2.5 5.60756e-05
3 1.84289e-05
3.5 4.21689e-06
4 9.64518e-07
};
%\addlegendentry{SCL - $R = \frac{1}{3}$}

\addplot[
    color=red,
    mark=pentagon,
    thick,
    mark size=3,
    dashed,
]
table {
1 0.102549
1.5 0.014084
2 0.00104555
2.5 0.000125188
3 1.61741e-05
3.5 2.33157e-06
4 2.31352e-07
};
%\addlegendentry{SCL - $R = \frac{1}{2}$}

\addplot[
    color=green,
    mark=diamond,
    thick,
    mark size=3,
    dashed,
]
table {
1 0.378191
1.5 0.180475
2 0.0328501
2.5 0.00242786
3 0.000149303
3.5 2.02782e-05
4 4.61012e-06
4.5 7.49834e-07
};
%\addlegendentry{SCL - $R = \frac{2}{3}$}

\addplot[
    color=black,
    mark=square,
    thick,
    mark size=3,
]
table {
1 0.00910353
1.5 0.00238706
2 0.000480811
2.5 0.000143751
3 3.12923e-05
3.5 6.61261e-06
4 2.13356e-06
4.5 3.72775e-07
};
%\addlegendentry{LPSCL - $R = \frac{1}{12}$}

\addplot[
    color=blue,
    mark=o,
    thick,
    mark size=3,
]
table {
1 0.00511941
1.5 0.000615294
2 8.79179e-05
2.5 1.58566e-05
3 3.09969e-06
3.5 7.77439e-07
4 1.34471e-07
};
%\addlegendentry{LPSCL - $R = \frac{1}{6}$}

\addplot[
    color=brown,
    mark=triangle,
    thick,
    mark size=3,
]
table {
1 0.0176762
1.5 0.00242933
2 0.000502549
2.5 6.60921e-05
3 1.88038e-05
3.5 4.21689e-06
4 9.64518e-07
4.5 1.70419e-07
};
%\addlegendentry{LPSCL - $R = \frac{1}{3}$}

\addplot[
    color=red,
    mark=pentagon,
    thick,
    mark size=3,
]
table {
1 0.106355
1.5 0.0150416
2 0.00130176
2.5 0.000150398
3 1.69001e-05
3.5 2.33157e-06
4 2.31352e-07
};
%\addlegendentry{LPSCL - $R = \frac{1}{2}$}

\addplot[
    color=green,
    mark=diamond,
    thick,
    mark size=3,
]
table {
1 0.379555
1.5 0.184599
2 0.0360519
2.5 0.00278358
3 0.000185854
3.5 1.98686e-05
4 4.61012e-06
4.5 7.49834e-07
5 9.55597e-08
};
%\addlegendentry{LPSCL - $R = \frac{2}{3}$}

\end{axis}
\end{tikzpicture}
  \ref{perf-legend1kL4}
  \caption{FER and BER performance comparison of decoding the 5G polar code of length $N=1024$ and $R \in \{\frac{1}{12},\frac{1}{6},\frac{1}{3},\frac{1}{2},\frac{2}{3}\}$, using LPSCL decoding with $L_{\max} = 4$ and $L_{10}=L_{9}=2$, and SCL decoding with $L=4$.}
  \label{fig:performance}
\end{figure*}

The channel LLR values are quantized with $4$ bits and internal LLR values with $6$ bits, with $2$ bits assigned to the fractional part, while PMs are quantized with $8$ bits \cite{hashemi_JETCAS}. The maximum node size is set to $16$ for Rate-0 and Rep nodes, and to $64$ for Rate-1 and SPC nodes.
Table~\ref{tab:HW} reports the area occupation and achievable frequency for the proposed decoder, and for the decoder based on the off-line identification technique, labelled as memory-based decoder. The two decoders differ in their implementation of the control unit (CU): its area occupation $A_{\text{CU}}$ in the proposed decoder is $24\%$ less than that of the memory-based decoder. This is due to the fact that the information computed off-line in the memory-based case, i.e. the equivalent of the \texttt{NodeType} signal, needs to be inserted in an FSM analogous to that used by the control unit of the proposed decoder. This FSM handles the node subtypes and the internal counters that determine when a special node decoding is terminated. Moreover, the memory-based case needs an additional information, \texttt{NodeStage}, to identify at which SC decoding tree stage the special node is encountered: the \texttt{NodeSize} information is derived from that. The \texttt{NodeStage} signal is inserted in its own FSM, that adds substantial complexity to the control unit, resulting in a larger $A_{\text{CU}}$. While the contribution of $A_{\text{CU}}$ to the total decoder area occupation $A_{\text{total}}$ is relatively small, with $A_{\text{total}}=1.410$~mm$^2$ and $A_{\text{total}}=1.454$~mm$^2$ for the proposed and the memory-based decoders respectively, the \texttt{NodeStage} FSM influences signals in the \texttt{NodeSize} and \texttt{NodeType} FSM, lengthening the critical path. In particular, the state of \texttt{NodeStage} is combined to the \texttt{NodeType} and \texttt{NodeSize} to determine the current and future node subtypes. This leads to a lower achievable frequency $f$, lower throughput $T_P$, and lower area efficiency $A_{\text{eff}}$ in the memory-based decoder in comparison with the proposed decoder, as provided in Table~\ref{tab:HW}.

\begin{table}
\centering
\caption{TSMC CMOS 65~nm synthesis results for $N=1024$, $P=4$, $L_{\max}=4$, and $L_{10}=L_{9}=2$.}
\label{tab:HW}
\setlength{\extrarowheight}{2.5pt}
\begin{tabular}{llcc}
\toprule

 & & Proposed & Memory-based\\
\midrule
 $A_{\text{CU}}$ & [$\mu$m$^2$] & $35881$ & $47025$\\
 $A_{\text{total}}$ & [mm$^2$] & $1.410$ & $1.454$\\		       
 $f$ & [MHz] & $955$ & $926$\\
 $T_P$ @ $R=\frac{1}{2}$ &[Mb/s] & $1223$ & $1186$ \\
$A_{\rm eff}$ @ $R=\frac{1}{2}$ & [$\frac{\rm Mb/s}{{\rm mm}^2}$] & $867$ & $816$\\
${\rm Mem}_{\text{ext}}$ & [bits] & $10240$ & $1025120$ \\
\bottomrule
\end{tabular}
\end{table} 

The proposed decoder fetches four required values of the relative reliability vector from memory, compares them with $K$, and identifies the node types efficiently. Table~\ref{tab:HW} also reports the external memory requirements $\text{Mem}_{\text{ext}}$ of the proposed decoder in comparison with the memory-based decoder considering 5G code rates are supported. For a code of length $1024$, the vector of relative reliabilities $\mathbf{v}$ contains $1024$ entries where each entry is stored with $10$ bits. Therefore, a total of $1024 \times 10 = 10240$ bits are stored in memory. For the memory-based decoder, the memory requirement is different for different values of $K$ (different rates). This is depicted in Fig.~\ref{fig:fastSchedule} where it can be seen that the list of operations is large for medium rates and becomes small as the rate becomes very high or very low. Note that the proposed decoder is capable of supporting any code rate within a given code length which is also foreseen in 5G \cite{3gpp_polarInfo}. If the memory-based decoder is designed such that it supports all the code rates of 5G for a code of length $1024$ ($12 \leq K \leq 1024$), the memory requirement of it considering a $4$-bit representation for \texttt{NodeType} and \texttt{NodeStage} signals is $128140\times 8 = 1025120$ bits, more than $100$ times larger than the number of bits required for the proposed decoder.

\begin{figure}
  \centering
  \begin{tikzpicture}
  \pgfplotsset{
    label style = {font=\fontsize{9pt}{7.2}\selectfont},
    tick label style = {font=\fontsize{7pt}{7.2}\selectfont}
  }

\begin{axis}[
	scale = 1,
	xmin = 0,
	xmax = 1024,
	ymin = 0,
    xlabel={$K$}, xlabel style={yshift=0.8em},
    ylabel={$\text{Mem}_{\text{ext}}$ [bits]}, ylabel style={yshift=-0.75em},
    grid=both,
    ymajorgrids=true,
    xmajorgrids=true,
    grid style=dashed,
%    width=0.5\columnwidth, height=7cm,
    thick,
]

\addplot[
    thick,
]
table {
0 16
1 16
2 144
3 128
4 128
5 128
6 128
7 128
8 128
9 144
10 112
11 144
12 192
13 192
14 176
15 176
16 176
17 160
18 224
19 224
20 208
21 208
22 208
23 208
24 208
25 208
26 288
27 288
28 272
29 288
30 288
31 288
32 240
33 336
34 336
35 320
36 336
37 336
38 304
39 304
40 304
41 416
42 448
43 448
44 448
45 464
46 448
47 416
48 416
49 416
50 400
51 400
52 400
53 432
54 416
55 416
56 432
57 480
58 480
59 448
60 448
61 576
62 576
63 576
64 576
65 608
66 608
67 608
68 592
69 576
70 560
71 576
72 576
73 576
74 624
75 624
76 624
77 592
78 592
79 592
80 592
81 624
82 624
83 608
84 608
85 672
86 672
87 688
88 688
89 672
90 608
91 608
92 656
93 640
94 640
95 656
96 656
97 672
98 672
99 672
100 672
101 672
102 656
103 624
104 624
105 624
106 576
107 576
108 608
109 672
110 672
111 672
112 688
113 672
114 656
115 656
116 624
117 672
118 688
119 688
120 688
121 688
122 688
123 688
124 768
125 768
126 768
127 768
128 800
129 752
130 736
131 752
132 752
133 768
134 752
135 736
136 800
137 800
138 800
139 768
140 768
141 768
142 736
143 752
144 784
145 784
146 768
147 768
148 768
149 768
150 768
151 800
152 816
153 896
154 944
155 896
156 896
157 896
158 880
159 864
160 896
161 912
162 880
163 880
164 880
165 880
166 880
167 880
168 880
169 864
170 832
171 816
172 816
173 800
174 832
175 832
176 832
177 880
178 880
179 880
180 944
181 960
182 976
183 976
184 976
185 960
186 960
187 960
188 1056
189 1056
190 1040
191 1040
192 1024
193 1024
194 1024
195 1024
196 1024
197 992
198 1024
199 976
200 976
201 992
202 992
203 992
204 992
205 976
206 976
207 960
208 960
209 960
210 976
211 1056
212 1056
213 1072
214 1040
215 1088
216 1088
217 1056
218 1056
219 1056
220 1056
221 1056
222 1056
223 1040
224 1072
225 1072
226 1088
227 1104
228 1104
229 1136
230 1232
231 1232
232 1216
233 1200
234 1200
235 1200
236 1168
237 1168
238 1168
239 1168
240 1088
241 1120
242 1104
243 1120
244 1072
245 1056
246 1056
247 1056
248 1056
249 1056
250 1056
251 1056
252 1072
253 1008
254 1008
255 1056
256 1056
257 1056
258 1056
259 1072
260 1088
261 1072
262 1072
263 1072
264 1072
265 1056
266 1104
267 1104
268 1136
269 1136
270 1152
271 1120
272 1120
273 1120
274 1088
275 1088
276 1088
277 1152
278 1104
279 1104
280 1088
281 1200
282 1232
283 1232
284 1248
285 1248
286 1248
287 1264
288 1248
289 1232
290 1168
291 1168
292 1168
293 1168
294 1184
295 1152
296 1152
297 1120
298 1104
299 1104
300 1104
301 1136
302 1136
303 1136
304 1152
305 1152
306 1200
307 1232
308 1232
309 1232
310 1216
311 1216
312 1184
313 1184
314 1168
315 1200
316 1184
317 1184
318 1184
319 1184
320 1184
321 1184
322 1136
323 1136
324 1200
325 1200
326 1216
327 1232
328 1232
329 1232
330 1216
331 1232
332 1200
333 1248
334 1248
335 1184
336 1216
337 1216
338 1216
339 1200
340 1168
341 1168
342 1168
343 1152
344 1152
345 1136
346 1152
347 1152
348 1136
349 1136
350 1184
351 1184
352 1200
353 1200
354 1200
355 1200
356 1232
357 1280
358 1280
359 1280
360 1264
361 1264
362 1216
363 1216
364 1184
365 1184
366 1200
367 1200
368 1264
369 1280
370 1280
371 1280
372 1264
373 1216
374 1232
375 1200
376 1200
377 1200
378 1200
379 1200
380 1200
381 1200
382 1200
383 1200
384 1264
385 1264
386 1264
387 1296
388 1312
389 1312
390 1280
391 1296
392 1296
393 1280
394 1312
395 1312
396 1328
397 1312
398 1344
399 1344
400 1344
401 1328
402 1328
403 1328
404 1328
405 1376
406 1376
407 1360
408 1440
409 1376
410 1376
411 1392
412 1360
413 1376
414 1376
415 1360
416 1312
417 1312
418 1312
419 1296
420 1280
421 1312
422 1312
423 1312
424 1296
425 1312
426 1280
427 1280
428 1264
429 1264
430 1264
431 1328
432 1328
433 1360
434 1328
435 1280
436 1280
437 1296
438 1296
439 1296
440 1344
441 1344
442 1344
443 1344
444 1344
445 1360
446 1376
447 1360
448 1344
449 1376
450 1392
451 1376
452 1376
453 1376
454 1376
455 1344
456 1344
457 1392
458 1344
459 1344
460 1344
461 1344
462 1312
463 1312
464 1312
465 1296
466 1312
467 1328
468 1328
469 1408
470 1440
471 1472
472 1472
473 1376
474 1376
475 1376
476 1376
477 1376
478 1328
479 1328
480 1296
481 1296
482 1296
483 1280
484 1296
485 1296
486 1296
487 1216
488 1216
489 1200
490 1184
491 1200
492 1264
493 1248
494 1248
495 1248
496 1248
497 1248
498 1248
499 1248
500 1280
501 1280
502 1248
503 1280
504 1264
505 1280
506 1264
507 1264
508 1264
509 1264
510 1264
511 1312
512 1312
513 1392
514 1344
515 1344
516 1328
517 1328
518 1328
519 1344
520 1360
521 1360
522 1376
523 1376
524 1360
525 1360
526 1328
527 1360
528 1360
529 1360
530 1360
531 1360
532 1344
533 1344
534 1360
535 1328
536 1376
537 1392
538 1328
539 1328
540 1328
541 1360
542 1360
543 1360
544 1360
545 1376
546 1376
547 1360
548 1376
549 1344
550 1344
551 1376
552 1360
553 1360
554 1360
555 1328
556 1328
557 1312
558 1344
559 1344
560 1360
561 1408
562 1392
563 1408
564 1504
565 1424
566 1424
567 1424
568 1424
569 1424
570 1424
571 1424
572 1408
573 1408
574 1360
575 1408
576 1424
577 1424
578 1424
579 1408
580 1408
581 1376
582 1376
583 1392
584 1344
585 1344
586 1376
587 1376
588 1376
589 1360
590 1344
591 1360
592 1360
593 1360
594 1360
595 1360
596 1296
597 1296
598 1296
599 1312
600 1312
601 1280
602 1264
603 1264
604 1280
605 1280
606 1280
607 1312
608 1312
609 1280
610 1328
611 1392
612 1408
613 1376
614 1376
615 1408
616 1408
617 1408
618 1424
619 1424
620 1424
621 1424
622 1408
623 1360
624 1376
625 1360
626 1344
627 1376
628 1376
629 1376
630 1360
631 1376
632 1376
633 1376
634 1312
635 1312
636 1280
637 1200
638 1200
639 1184
640 1216
641 1216
642 1216
643 1216
644 1216
645 1216
646 1184
647 1232
648 1248
649 1264
650 1264
651 1264
652 1264
653 1296
654 1312
655 1312
656 1312
657 1312
658 1296
659 1232
660 1216
661 1216
662 1216
663 1232
664 1232
665 1232
666 1232
667 1232
668 1280
669 1280
670 1280
671 1232
672 1216
673 1232
674 1232
675 1232
676 1200
677 1152
678 1152
679 1168
680 1168
681 1136
682 1200
683 1216
684 1216
685 1248
686 1232
687 1232
688 1232
689 1216
690 1216
691 1200
692 1200
693 1232
694 1232
695 1248
696 1248
697 1248
698 1248
699 1200
700 1136
701 1136
702 1152
703 1152
704 1120
705 1104
706 1120
707 1168
708 1200
709 1200
710 1200
711 1216
712 1216
713 1216
714 1216
715 1216
716 1184
717 1152
718 1136
719 1120
720 1120
721 1120
722 1120
723 1152
724 1152
725 1184
726 1136
727 1200
728 1200
729 1200
730 1200
731 1200
732 1200
733 1200
734 1216
735 1200
736 1216
737 1232
738 1232
739 1232
740 1232
741 1216
742 1264
743 1232
744 1232
745 1216
746 1216
747 1232
748 1264
749 1200
750 1200
751 1200
752 1168
753 1200
754 1200
755 1184
756 1184
757 1184
758 1136
759 1152
760 1152
761 1136
762 1120
763 1120
764 1120
765 1120
766 1136
767 1136
768 1200
769 1200
770 1200
771 1152
772 1152
773 1152
774 1152
775 1168
776 1136
777 1136
778 1136
779 1136
780 1152
781 1152
782 1136
783 1168
784 1184
785 1184
786 1184
787 1200
788 1280
789 1280
790 1248
791 1248
792 1248
793 1280
794 1168
795 1168
796 1168
797 1168
798 1152
799 1104
800 1152
801 1184
802 1152
803 1168
804 1152
805 1152
806 1136
807 1152
808 1152
809 1136
810 1136
811 1136
812 1120
813 1136
814 1136
815 1104
816 1104
817 1104
818 1104
819 1104
820 1104
821 1008
822 1056
823 1056
824 1056
825 1056
826 1088
827 1088
828 1104
829 1104
830 1088
831 1088
832 1104
833 1056
834 1056
835 976
836 976
837 976
838 976
839 992
840 976
841 944
842 976
843 992
844 992
845 992
846 1008
847 912
848 912
849 912
850 912
851 912
852 912
853 912
854 880
855 912
856 896
857 896
858 896
859 912
860 912
861 848
862 864
863 864
864 864
865 784
866 800
867 848
868 848
869 848
870 832
871 816
872 816
873 816
874 768
875 784
876 752
877 752
878 752
879 784
880 784
881 784
882 784
883 800
884 816
885 816
886 800
887 816
888 752
889 800
890 800
891 800
892 768
893 736
894 736
895 768
896 800
897 816
898 816
899 816
900 816
901 800
902 800
903 720
904 704
905 720
906 720
907 704
908 656
909 672
910 672
911 720
912 720
913 720
914 720
915 720
916 704
917 704
918 640
919 608
920 640
921 640
922 656
923 720
924 720
925 720
926 720
927 736
928 736
929 720
930 704
931 656
932 672
933 672
934 672
935 672
936 608
937 608
938 608
939 624
940 592
941 592
942 592
943 624
944 624
945 624
946 608
947 608
948 608
949 624
950 640
951 592
952 592
953 592
954 624
955 624
956 624
957 624
958 624
959 624
960 624
961 624
962 592
963 544
964 544
965 560
966 544
967 544
968 544
969 512
970 512
971 512
972 528
973 560
974 560
975 560
976 560
977 592
978 608
979 592
980 592
981 464
982 432
983 448
984 448
985 448
986 432
987 320
988 320
989 368
990 368
991 352
992 352
993 368
994 368
995 368
996 368
997 368
998 384
999 288
1000 288
1001 288
1002 288
1003 304
1004 304
1005 304
1006 224
1007 160
1008 160
1009 160
1010 176
1011 192
1012 144
1013 176
1014 144
1015 144
1016 128
1017 128
1018 128
1019 128
1020 128
1021 128
1022 144
1023 16
1024 16
};

\end{axis}
\end{tikzpicture}
  \caption{Memory requirements to store the list of operations of the memory-based decoder for different values of $K$. The polar code of length $1024$ is used which is adopted in 5G \cite{3gpp_polarSequence}.}
  \label{fig:fastSchedule}
\end{figure}
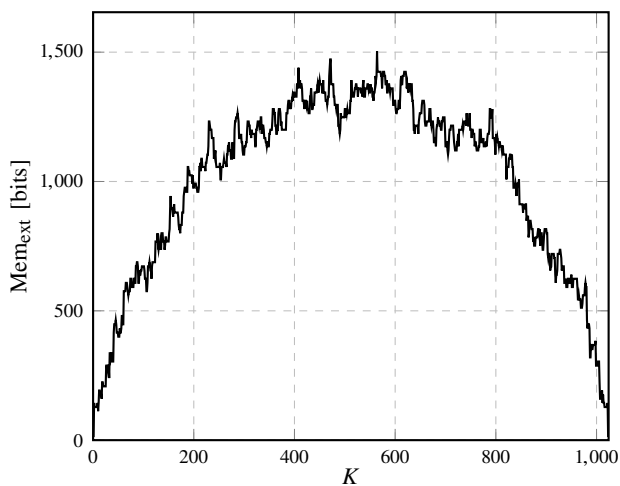

Artisan dual-port SRAM compiler was used for the implementation of the external memories. Table~\ref{tab:SRAM} shows the area occupation of the external memory for the proposed decoder in comparison with the memory-based decoder. While the proposed decoder supports all the code rates, the memory requirement of the memory-based decoder depends on the number of code rates it can support. In Table~\ref{tab:SRAM}, we showed four cases of memory requirements for the memory-based decoder: when it supports $5$ code rates of $\{\frac{1}{12},\frac{1}{6},\frac{1}{3},\frac{1}{2},\frac{2}{3}\}$, when it supports $10$ code rates of $\{\frac{1}{16},\frac{1}{12},\frac{1}{8},\frac{1}{6},\frac{1}{4},\frac{1}{3},\frac{1}{2},\frac{2}{3},\frac{5}{6},\frac{7}{8}\}$, when it supports $20$ code rates of $\{\frac{1}{24},\frac{1}{16},\frac{1}{12},\frac{1}{8},\frac{1}{6},\frac{1}{5},\frac{1}{4},\frac{5}{16},\frac{1}{3},\frac{3}{8},\frac{2}{5},\frac{1}{2},\frac{3}{5},\frac{5}{8},\frac{2}{3},\frac{11}{16},\frac{3}{4},\frac{4}{5},\frac{5}{6},\frac{7}{8}\}$, and when it supports all the code rates considered in 5G, similar to the proposed decoder. It can be seen that the proposed decoder occupies a smaller area in comparison with the memory-based decoder even when the memory-based decoder supports only $5$ code rates. In fact, the area occupation of the memory-based decoder increases as the number of supported code rates increases. This consequently reduces the area efficiency of the memory-based decoder as can be seen in Table~\ref{tab:SRAM}. The area occupation of the proposed decoder is only $38\%$ of that of the memory-based decoder when both decoders support all 5G code rates.

\begin{table}[t]
\centering
\caption{SRAM synthesis results for external memories.}
\label{tab:SRAM}
\setlength{\extrarowheight}{2.5pt}
\begin{tabular}{llccc}
\toprule

 & & $\text{Mem}_{\text{ext}}$ & $A_{\text{total}}$ & $A_{\text{eff}}$ @ $R=\frac{1}{2}$ \\
 & & [mm$^2$] & [mm$^2$] & [$\frac{\rm Mb/s}{{\rm mm}^2}$] \\
\midrule
Proposed & All rates & $0.039$ & $1.449$ & $844$ \\
\cmidrule(lr){1-2}
\cmidrule(lr){3-5}
\multirow{4}{*}{Memory-based} & $5$ rates & $0.033$ & $1.487$ & $798$ \\
 & $10$ rates & $0.047$ & $1.501$ & $790$ \\
 & $20$ rates & $0.080$ & $1.534$ & $773$ \\
 & All rates & $2.358$ & $3.812$ & $311$ \\
\bottomrule
\end{tabular}
\end{table} 

It is worth mentioning that the goal of this paper is to propose a low-complexity approach to generate the list of operations for fast SC-based decoders directly on hardware, therefore, allowing for the implementation of a fast and rate-flexible SC-based decoder. Our implementation results show that by using the proposed method, there is a negligible area occupation overhead or throughput loss in comparison with the memory-based decoders, while having a completely rate-flexible decoder.

\begin{table*}
	\centering
	\caption{Comparison with state-of-the-art decoders.}
	\label{tab:compare}
		\setlength{\extrarowheight}{1.7pt}
%\scriptsize
	\begin{tabularx}{0.95\textwidth}{lYYYYYY}
\toprule
 & This work & \cite{hashemi_FSSCL_TSP} & \cite{hashemi_SSCL_TCASI} & \cite{yuan_multibit_LLR} & \cite{xiong_symbol}\textsuperscript{\textdagger} & \cite{lin_SCL}\textsuperscript{\textdagger} \\

\cmidrule(lr){1-1}
\cmidrule(lr){2-2}
\cmidrule(lr){3-3}
\cmidrule(lr){4-4}
\cmidrule(lr){5-5}
\cmidrule(lr){6-6}
\cmidrule(lr){7-7}

$A_{\text{total}}$ [mm$^2$] & $1.449$ & $1.797$ ($4.155$) & $1.22$ ($3.578$) & $0.62$ & $0.73$ & $2.00$ \\

$f$ [MHz] & $955$ & $840$ & $961$ & $498$ & $692$ & $558$ \\

$T_P$ [Mb/s] & $1223$ & $1338$ & $1146$ & $935$ & $551$  & $1578$ \\

Latency [$\upmu$s] & $0.84$ & $0.77$ & $0.89$ & $1.10$ & $1.86$  & $0.66$ \\

$A_{\rm eff}$ [Mb/s/mm$^2$] & $844$ & $744$ ($322$) & $939$ ($320$) & $1508$ & $755$ & $789$ \\
\midrule[\heavyrulewidth]
\multicolumn{7}{l}{\textsuperscript{\textdagger}\footnotesize{The results are originally based on TSMC 90~nm technology and are scaled to TSMC 65~nm technology.}} \\
\bottomrule
	\end{tabularx}
\end{table*}

The main advantage of the proposed approach is that given the design code length, any code with the same $N$ can be decoded using the Fast-SSCL-SPC algorithm without foreknowledge of the information/frozen bit sequence, regardless of rate and target $E_b/N_0$. On the contrary, the memory-based decoder needs to store the \texttt{NodeType} and \texttt{NodeStage} information for each considered code in an external memory of ${\rm Mem}_{\text{ext}}$ bits.
%Considering $N=1024$, the \texttt{NodeType} and \texttt{NodeStage} signals require a $4$-bit representation. The last rows of Table~\ref{tab:HW} report the external memory requirements ${\rm Mem}_{\text{ext}}$ for each considered code rate: while the proposed identification method does not require any external memory, the memory-based decoder requires thousands of memory bits for each code rate, with a total of $6824$ bits for the rates considered in 5G. Implemented in TSMC 65~nm CMOS technology with registers, the area occupation of such external memory would be $0.717$~mm$^2$.

Table~\ref{tab:compare} compares the proposed decoder to other architectures in the state of the art which use $64$ parallel PEs. Results are reported for $\mathcal{P}(1024,512)$ and $L=4$. The architectures presented in \cite{hashemi_FSSCL_TSP} and \cite{hashemi_SSCL_TCASI} are based on the Fast-SSCL-SPC and SSCL-SPC algorithms, respectively: it is possible to add the cost of the external memory directly to their area occupation and evaluate its impact on the area efficiency, considering all the code rates in 5G are supported. These modified results are reported within parentheses. It can be seen that the external memory increases $A_{\text{total}}$ by $131\%$ in \cite{hashemi_FSSCL_TSP} and by $193\%$ in \cite{hashemi_SSCL_TCASI}: the proposed special node identification technique is thus able to substantially limit the area occupation and increase the area efficiency in both architectures. The architecture presented in this work has higher $A_{\rm eff}$ and lower $A_{\text{total}}$ than both \cite{hashemi_SSCL_TCASI} and \cite{hashemi_FSSCL_TSP}. Different design choices in terms of concurrent operations in the special nodes lead to a slightly lower $T_P$ than \cite{hashemi_FSSCL_TSP}, together with a substantially lower $A_{\text{total}}$ and higher $A_{\rm eff}$.

The architectures presented in \cite{yuan_multibit_LLR,xiong_symbol,lin_SCL} do not rely on a special-node-based decoding algorithm: thus, the throughput benefits and complexity saving of the proposed node identification technique cannot be directly evaluated. Moreover, the synthesis results of \cite{yuan_multibit_LLR} were reported in 90~nm technology, but they were carried out in 65~nm technology. Therefore, a factor of $90/65$ was used to convert the frequency, and a factor of $\left(65/90\right)^2$ was used to convert the area of the decoder from 90~nm to 65~nm technology in \cite{yuan_multibit_LLR}. The same conversion factors were used to convert to 65~nm technology the synthesis results in \cite{xiong_symbol,lin_SCL}, which were synthesized with a 90~nm node.

Our work shows $31\%$ higher throughput and $31\%$ lower latency with respect to the multibit decision SCL decoder architecture of \cite{yuan_multibit_LLR}, while the smaller area occupation of \cite{yuan_multibit_LLR} leads to a higher $A_{\rm eff}$. The decoder in \cite{xiong_symbol} shows lower area occupation than our work. However, the architecture proposed in this work achieves $122\%$ higher throughput and $55\%$ lower latency, leading to $11\%$ higher area efficiency. The high throughput SCL decoder architecture of \cite{lin_SCL} achieves higher throughput and lower latency than this work, at the cost of $38\%$ higher area occupation and $7\%$ lower $A_{\rm eff}$. Moreover, \cite{lin_SCL} relies on tunable parameters that can lead to more than $0.2$~dB error-correction performance loss. These parameters also reduce the flexibility of the decoder, since for each code rate, a different set of parameters needs to be used. However, the decoder proposed in this paper is designed to guarantee rate-flexibility, making it suitable for 5G applications.

\section{Conclusion} \label{sec:conc}

The main drawback of the fast successive-cancellation-based decoders for polar codes is that they require to store a list of operations for each code rate in a dedicated memory, in order to tell the decoder when a special node in a polar code graph is reached. In this paper, we tackled this issue by proposing a technique to generate the list of operations on-the-fly directly in hardware. We proved that this technique can be applied to polar codes of any rate, therefore, removing the memory needed to store the list of operations completely. We proposed a hardware architecture for the proposed technique and showed that the total area occupation of the proposed decoder is $38\%$ of the base-line memory-based decoder, if 5G code rates are considered.

\section*{Acknowledgments}

The authors would like to thank Arash Ardakani and Harsh Aurora of McGill University for helpful discussions. S.~A.~Hashemi is supported by a Postdoctoral Fellowship from the Natural Sciences and Engineering Research Council of Canada (NSERC). M.~Mondelli is supported by an Early Postdoc.Mobility fellowship from the Swiss National Science Foundation and by the Simons Institute for the Theory of Computing.

% Generated by IEEEtran.bst, version: 1.12 (2007/01/11)

\end{document}